\documentclass{lipics-v2021}
\nolinenumbers
\usepackage[T1]{fontenc}
\usepackage{babel}
\usepackage{verbatim}
\usepackage{amsmath}
\usepackage{amsthm}
\usepackage{amssymb}
\usepackage{mathtools}
\usepackage{float}
\usepackage{graphicx}
\usepackage{caption}
\usepackage{subcaption}
\usepackage{xcolor}
\usepackage{hyperref}
\hypersetup{
    colorlinks,
    linkcolor={red!50!black},
    citecolor={blue!50!black},
    urlcolor={blue!80!black}
}

\makeatletter

\usepackage{amstext}
\usepackage[inline]{enumitem}
\usepackage{todonotes}
\usepackage{thmtools}
\usepackage{thm-restate}
\usepackage[capitalize]{cleveref}
\usepackage{graphicx}

\usepackage{algorithm}
\usepackage[noend]{algpseudocode}

\global\long\def\R{\mathbb{R}}%
\global\long\def\R{\mathbb{R}}%
\global\long\def\Z{\mathbb{Z}}%

\global\long\def\eps{\varepsilon}%
\newcommand{\NP}{\mathrm{NP}} 
\newcommand{\DTIME}{\mathrm{DTIME}}

\DeclareMathOperator*{\argmax}{arg\;max}

\DeclareMathOperator{\poly}{poly}

\DeclareMathOperator{\Dyn}{Dyn}
\DeclareMathOperator{\ext}{ext}
\DeclareMathOperator{\interior}{int}

\DeclareMathOperator{\pr}{Pr}

\DeclareMathOperator{\OPT}{OPT} 
\DeclareMathOperator{\CHA}{ACC}
\newcommand{\cP}{\CHA}
\DeclareMathOperator{\ttop}{top}
\DeclareMathOperator{\bbot}{bot}

\newcommand{\calb}{\mathcal{B}} 
\newcommand{\calc}{\mathcal{C}} 
\newcommand{\cald}{\mathcal{D}} 
\newcommand{\calg}{\mathcal{G}} 
\newcommand{\cali}{\mathcal{I}} 
\newcommand{\call}{\mathcal{L}} 
\newcommand{\calo}{\mathcal{O}} 
\newcommand{\calp}{\mathcal{P}} 
\newcommand{\calr}{\mathcal{R}} 
\newcommand{\calv}{\mathcal{V}} 

\newcommand{\set}[1]{\left\{\, #1 \,\right\}} 

\ifdefined\DEBUG
  \newcommand{\ant}[1]{\textcolor{red}{#1}}
  \newcommand{\fab}[1]{\textcolor{orange}{#1}}
  \newcommand{\edi}[1]{\textcolor{blue}{#1}}

  \def\rem#1{{\marginpar{\raggedright\scriptsize #1}}}
  \newcommand{\anote}[1]{\rem{\textcolor{red}{$\bullet$ #1}}}
  \newcommand{\fabr}[1]{\rem{\textcolor{orange}{$\bullet$ #1}}}
  \newcommand{\enote}[1]{\rem{\textcolor{blue}{$\bullet$ #1}}}
  \newcommand{\mnote}[1]{\rem{\textcolor{magenta}{$\bullet$ #1}}}
\else
  \newcommand{\ant}[1]{#1}
  \newcommand{\fab}[1]{#1}
  \newcommand{\edi}[1]{#1}

  \newcommand{\anote}[1]{}
  \newcommand{\fabr}[1]{}
  \newcommand{\enote}[1]{}
  \newcommand{\mnote}[1]{}
\fi

\title{Approximating the Maximum Independent Set of Convex Polygons with a Bounded Number of Directions}
\titlerunning{Approximating MISP with a bounded number of directions}

\author{Fabrizio Grandoni}{IDSIA, USI-SUPSI, Switzerland \and \url{https://people.idsia.ch/~grandoni}}{fabrizio@idsia.ch}{}{}

\author{Edin Husi\'{c}}{IDSIA, USI-SUPSI, Switzerland \and \url{https://zhero9.github.io}}{edin.husic@idsia.ch}{}{}

\author{Mathieu Mari}{LIRMM, University of Montpellier, CNRS, Montpellier, France \and \url{https://mimuw.edu.pl/~mmari}}{mathieu.mari@lirmm.fr}{}{}

\author{Antoine Tinguely}{IDSIA, USI-SUPSI, Switzerland}{antoine.tinguely@idsia.ch}{}{}

\authorrunning{F. Grandoni, E. Husi\'{c}, M. Mari and A. Tinguely}

\Copyright{Fabrizio Grandoni, Edin Husi\'{c}, Mathieu Mari and Antoine Tinguely} 


\begin{document}
 
\maketitle

\begin{abstract}
    In the maximum independent set of convex polygons problem, we are given a set of $n$ convex polygons in the plane with the objective of selecting a maximum {cardinality} subset of non-overlapping polygons. {Here we study a special case of the problem where the edges of the polygons can take at most $d$ fixed directions.} 
    We present an $8d/3$-approximation algorithm for this 
    problem running in time $O((nd)^{O(d4^d)})$. The previous-best polynomial-time approximation (for constant $d$) was a classical $n^\eps$ approximation by Fox and Pach [SODA'11] \ant{that has recently been improved to a $\OPT^{\varepsilon}$-approximation algorithm by Cslovjecsek, Pilipczuk and W{\k{e}}grzycki [SODA '24], which also extends to an arbitrary set of convex polygons.}
    
    Our result builds on, and generalizes the recent constant factor approximation algorithms for the maximum independent set of axis-parallel rectangles problem {(which is a special case of our problem with $d=2$)} by Mitchell~{[FOCS'21]} and G\' alvez, Khan, Mari, M\" omke, Reddy, and Wiese~{[SODA'22]}.
\end{abstract}

\newpage
\setcounter{page}{1} 

\section{Introduction}
\label{sec:introduction}

The Maximum Independent Set of Convex Polygons problem (MISP) is a natural geometric packing 
problem with many applications in map labeling~\cite{de2000applications,verweij1999optimisation}, cellular networks~\cite{malesinska1997graph}, unsplittable flow~\cite{bonsma2014constant}, chip manufacturing~\cite{hochbaum1985approximation}, or data mining~\cite{fukuda2001data,lent1997clustering}. 
Given a set of $n$ convex polygons in the plane, the goal is to select a maximum number of them  such that the polygons are pairwise non-overlapping.

MISP is NP-hard~\cite{fowler1981optimal,imai1983finding}, hence it makes sense to design approximation algorithms for it. \ant{Disappointingly, the best (polynomial-time) approximation ratio for MISP (more precisely for $k$-intersecting curves) has been $n^{\eps}$~\cite{fox2011computing}, for any fixed constant $\eps>0$. This ratio has recently been improved to $\OPT^{\varepsilon}$~\cite{cslovjecsek2024polynomial}.} 
\enote{Mention that up to a constant, the most difficult case is the MIS of lines?}

\smallskip\noindent \textbf{Approximation Schemes.} Interestingly, there is a quasi-polynomial time approximation scheme (QPTAS) for MISP~\cite{AdamaszekHW19} \ant{(if the polygons have at most quasi-polynomially many vertices in total)}.\fabr{Maybe ~\cite{AdamaszekHW19} works only for a constant number of edges per polygon: this is not clear here \\
A: It works if there are at most quasi-polynomially many vertices in total. But somehow it's not super important, how do you encode quasi-polynomially many vertices?}
\enote{Copied from~\cite{AdamaszekHW19}: "In particular, our algorithm works for axis-parallel rectangles, line segments, and arbitrary polygons." So, they don't even need convex.}Thus, the problem is \emph{not} APX-Hard, assuming $\NP\nsubseteq \DTIME(2^{\textrm{polylog}(n)})$, suggesting 
that it should be possible to obtain a polynomial time approximation scheme (PTAS) for the problem.

If we assume that we are allowed to shrink the polygons by a factor $1-\delta$ for an arbitrarily small constant $\delta$, then there is a PTAS for the problem~\cite{wiese2018independent}. Note that here the output is compared to the optimal solution without shrinking.

When the input polygons are fat, e.g.,  regular polygons, 
then PTASes are known~\cite{chan2003polynomial,erlebach2005polynomial}. 

\smallskip\noindent \textbf{Axis-parallel rectangles.}
A prominent special case of MISP that has attracted a lot of attention over the years is the maximum independent set of axis-parallel rectangles (MISR), 
where all the polygons are rectangles with their edges parallel with the axes. 
An $O(\log n)$ approximation for MISR was given in \cite{khanna1998approximating,nielsen2000fast}. This was slightly improved to $O(\log n/\log\log n)$ in \cite{ChanHarPeled2012}, and substantially improved to $O(\log\log n)$ in \cite{CC2009}. In a recent breakthrough result,
\ant{Mitchell~\cite{mitchell2021approximating} presented the first constant factor approximation algorithm with approximation ratio $10$, and later $3+\eps$ in an updated version \cite{mitchell2022approximating} with a considerably shorter case analysis.}
Subsequently, his approach was simplified and improved to a $(2+\eps)$-approximation algorithm by 
G\' alvez, Khan, Mari, M\" omke, Reddy,
and Wiese~\cite{galvezSODA, galvez20212+}.
These approaches rely on a dynamic program that considers all the partitions of a bounding box containing the instance into a number of containers with constant complexity (constant number of line segments).

\smallskip \noindent \textbf{Our contribution.}
\fab{With the goal of better understanding the approximability of MISP, in this paper, we consider the following natural special case of MISP: $d$-MISP is the special case of MISP where the edges of the input polygons are parallel to a given set $\cald$ of $d=|\cald|$ directions. \edi{Notice that MISR is equivalent to $2$-MISP.} 
Our main result is a constant approximation for $d$-MISP when $d$ is a constant.}
\enote{Orthogonality does not matter for $d=2$, it should be wlog. \fab{I remember that we spoke about it, but we had some doubts}}

\begin{theorem}
\label{thm:main}
There exists an $8d/3$-approximation algorithm for \fab{$d$-MISP} running in time \ant{$O((nd)^{O(d4^d)})$}.
\end{theorem}
\fab{Our result builds on the approaches in \cite{galvezSODA, galvez20212+,mitchell2022approximating}, however we have to face several additional complications. In particular, already for $d=3$ the algorithm and its analysis deviates substantially from the known (polynomial-time) results in the literature about axis-aligned rectangles. An overview of our approach is given in Section \ref{sec:approach}.}

\smallskip \noindent \textbf{Related Work.}
One can consider a natural weighted version of MISP,\fabr{Notice that MISP contains "problem", so we should not say MISP problem or so} where each convex polygon has a positive weight, and the goal is to find an independent set of maximum total weight. The weighted version of MISR was studied in the literature, and the current-best polynomial time approximation factor is $O(\log\log n)$ \cite{ChalermsookW21}. 
We remark that our approach, likewise the approaches in \cite{galvezSODA, galvez20212+, mitchell2021approximating}, does not seem to extend to the weighted case. In particular, finding a constant approximation for weighted MISR remains a challenging open problem. We remark that the QPTAS in \cite{AdamaszekHW19} extends to the weighted case, hence suggesting that the weighted version of MISP might also admit a PTAS.

MISR was also studied in terms of parameterized algorithms. Marx \cite{marx2005efficient} proved that the problem is W[1]-hard, which rules out the existence of an EPTAS. A parameterized approximation scheme for the problem is given in \cite{GKW19}.

A rectangle packing problem related to MISR is the 2D Knapsack problem. Here we are given an axis-parallel square (the \emph{knapsack}) and a collection of axis-parallel rectangles. 
The goal is to pack a maximum cardinality (or weight) subset of rectangles in the knapsack (without rotations). 2D Knapsack admits a QPTAS \cite{AW15}
\anote{if data is quasi-polynomially bounded}
and a few constant approximation algorithms are known \cite{GGIHKW21,GalSocg21,jansen2004rectangle}. Here as well, finding a PTAS is a challenging open problem.

Bonsma et al. \cite{bonsma2014constant} established an intriguing connection between MISR and the Unsplittable Flow on a Path problem. A PTAS for the latter problem was recently obtained \cite{GMW22}, closing a very long line of research (see, e.g., \cite{AGLW14,BCES06,BFKS14,bonsma2014constant,GMW22soda,GMWZ18}).



\section{Preliminaries}
\label{sec:preliminaries}
In this paper, a \ant{(possibly closed)} \emph{curve} is always assumed to be a polygonal chain (or a singleton \ant{point}) and a \emph{polygon $S$} is a bounded set with non-empty interior $\interior(S)$
and whose  \emph{boundary $\partial S$} is a \ant{closed} curve. We denote the \emph{closure of $S$} as $\bar{S}$, so $\bar S = \partial S \cup \interior(S)$.
We say that two polygons $S, T$ (with non-empty interior) \emph{touch} if $\interior(S) \cap \interior(T) = \emptyset$ but $\partial S \cap \partial T \neq \emptyset$ and \emph{intersect} if $\interior(S) \cap \interior(T) \neq \emptyset$.
A curve $f$ \emph{touches} $S$ if $f \cap \interior(S) = \emptyset$ but $f \cap \partial S \neq \emptyset$.

A line segment or curve is called \emph{degenerate} if it is a singleton \ant{point}.\fabr{"singleton" means a point? \ant{yes} Is "edge" just a segment? \ant{yes, an edge is a line segment on the boundary of a polygon (at least that's how we use it)}}
A line segment or curve is assumed to be non-degenerate unless we explicitly state the opposite.
For an (oriented) line segment $e=\overline{ww'}$
(resp. curve $\gamma = w_1w_2 \cdots w_k$) we define the \emph{head of $e$ (of $\gamma$)} as $h(e) = w'$ ($h(\gamma) = w_k$) and the \emph{tail of $e$ (of $\gamma$)} as $t(e) = w$ ($t(\gamma) = w_1$) and the \emph{interior of $e$ (of $\gamma$)} as $\interior(e) = e \setminus\{h(e), t(e)\}$ ($\interior(\gamma) = \gamma \setminus\{h(\gamma), t(\gamma)\}$).
For a degenerate line segment (resp.\ curve), 
the head and the tail coincide with the line segment (resp.\ curve).

For a vector $v=(x, y)$, let $v^\perp \coloneqq (y, -x)$ (which is $v$ rotated clockwise orthogonally).
\ant{For a positive integer $k$, let $[k] \coloneqq \{1, \dots, k\}$.}

\begin{figure}%
    \centering
    \includegraphics[width=0.6\textwidth]{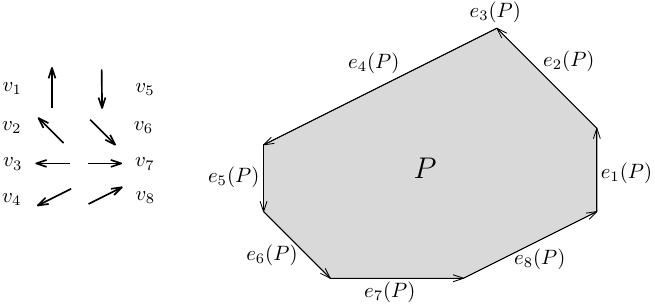}%
    \caption{A convex polygon in $4$ directions. The edge $e_3(P)$ is degenerate.}%
    \label{fig:convexPolygon}%
\end{figure}

\noindent\textbf{Input.} For a fixed positive integer $d$, the input of our problem is given by a set of (pairwise linearly independent) \emph{$d$ direction defining vectors} $\cald = \{v_1, \dots, v_{d}\} \subseteq \Z^2$ and a set $\cali$ of $n$ convex polygons with edges oriented along the directions given in $\cald$.
\ant{Polygons of this type are sometimes called \emph{$d$-discrete orientation polytopes ($d$-DOPs)} \cite{klosowski1998efficient}. In this paper, we will more casually refer to them as (input) polygons; the significance of the word ``polygon'' will be clear from context.}
Without loss of generality, assume $v_1 = (0,1)$ and that $v_2, \dots, v_d$ point to the left and are ordered by decreasing slope, see Figure~\ref{fig:convexPolygon}.
For $i \in \{d+1, \dots, 2d\}$, let $v_i \coloneqq -v_{i-d}$. \edi{The indices of the directions are counted modulo $2d$, i.e., $i = i + 2d = i-2d$.}
More explicitly, each polygon $P \in \cali$ is \ant{encoded by $2d$ integers
$p_1(P), \dots, p_{2d}(P)$} as $P = \{x\in \R^2:\ x^\intercal v^\perp_i < p_i(P),\ \forall i \in [2d]\}$; and thus $\bar P = \{x\in \R^2:\ x^\intercal v^\perp_i \le p_i(P),\ \forall i \in [2d]\}$.
We assume\fabr{$\bar{P}$ was defined for $P$ a curve, right?} that those linear inequalities are all tight, including redundant ones%
\footnote{\edi{An inequality is \emph{redundant} if we can remove it from the definition of $P$ without affecting $P$.}}%
, i.e.,\ $e_i(P) \coloneqq \bar{P} \cap \{x:\ x^\intercal v^\perp_i = p_i(P)\} \neq \emptyset$ for every $i \in [2d]$. \edi{$e_i(P)$ is called the \emph{edge of $P$ in direction $v_i$}}.
Then, for every $i \in [2d]$, $e_i(P)$ and $e_{i+1}(P)$ are incident and $h(e_i(P)) = t(e_{i+1}(P))$. 
Note that $e_1(P)e_2(P) \cdots e_{2d}(P)$ forms a positively oriented  closed curve.

\noindent\textbf{Grid.}
Let $\call_1$ be the set of all lines in directions $v_1, \dots, v_d$ passing through the vertices of the input polygons. In particular, all the edges (including the degenerate ones) of all the polygons in the input lie on the lines in $\call_1$. Notice that $|\call_1|\leq 2d^2n$.
We recursively define $\calv_k$, for $k \in [2d]$ and $\call_k$, for $k\in \{2, \dots, 2d\}$ as follows:
$\calv_k$ is the set of intersection points of any two (non-parallel) lines in $\call_k$,
and $\call_k$ is the set of all lines in directions $\cald$ passing through points in $\calv_{k-1}$.
We define the grid $\calg_k = (\call_{k}, \calv_{k})$.
Since $|\calv_k| \le |\call_{k}|^2$ and $|\call_k| \le |\calv_{k-1}|\cdot d$, it follows that $|\calv_k| \le (2d^3n)^{2^k}$.
The grid $\calg_{2d}$ form the coordinate system of our algorithm: every geometric object appearing in the algorithm and the analysis lies on $\calg_{2d}$.
A line segment \emph{$s$ lies on $\calg_k$} 
if $s$ lies on some line in $\call_{k}$ and the extreme points of $s$ lie on $\calv_k$. Similarly, a curve or polygon lies on $\calg_k$ if all of its line segments do so.

\noindent\textbf{Container.} Consider the grid $\calg_{1}$. Let $C^* \in \calg_{1}$ be a parallelogram that encloses all polygons in $\cali$; we call $C^{*}$ the \emph{bounding box}.%
\footnote{It can, for example, be chosen as a parallelogram delimited by the leftmost and rightmost vertical lines and the top and bottom $v_2$-oriented lines in $\calg_1$
(i.e., the extension of $e_2(P')$ where $P' = \argmax_{P \in \cali} p_2(P)$ and the extension of $e_{d+2}(P'')$ where $P'' = \argmax_{P \in \cali} p_{d+2}(P)$).}
A \emph{container} (see Figure \ref{fig:weaklySimple}(a)) is a polygon on $\calg_{2d}$ with positively oriented boundary $s_1f_1s_2f_2 \dots s_\kappa f_\kappa$ where $2 \leq \kappa \leq 5$, such that:
\begin{itemize}
    \item $s_1, s_2, \dots, s_\kappa$ are disjoint and possibly degenerate \emph{parallel} line segments on $\calg_{2d}$ (these will later be called \emph{cutting lines}).
    \item For all $j\in [\kappa]$, $f_j$ is a simple curve on $\calg_{2d}$ consisting of at most $2d+1$ line segments and $t(f_j) = h(s_j)$ and $h(f_j) = t(s_{j+1})$ for every $j\in [\kappa]$ (where $s_{\kappa + 1} = s_1$).
    \item For all $j\in [\kappa]$, $\interior(s_j)$ does not intersect with any other part of the boundary of the container.
    \item For all $i,j \in [\kappa], i \neq j$, the curves $f_i$ and $f_j$ might touch but do not cross (defined below).
\end{itemize} 
In particular, a container has at most $10d+10$ \edi{line segments}.
Let $\calc$ be the set of all containers $C$ with $\interior(C) \subseteq \interior(C^*)$.
In particular, $C^{*}$ is a container and $C^{*}\in \calc$.
A \emph{bipartition of $C \in \calc$} is a pair $\{C_1, C_{2}\} \subseteq \calc$ such that $C_1, C_2$ split up $C$, i.e., $\interior(C)\setminus(\partial C_1 \cup \partial C_2) = \interior(C_1) \cup \interior(C_2)$ and $C_1$ and $C_2$ may touch but not intersect.

\noindent\textbf{Crossing curves.} Two curves \emph{cross} (see also Figure \ref{fig:weaklySimple}(b)) if each one of them contains  a connected subcurve $w_0 w_1 \cdots w_{k}$ and $q_0 q_1 \cdots q_{k}$, respectively, which form a \emph{crossing}, i.e., if $w_0 \neq q_0$, $w_{k} \neq q_{k}$, $w_i = q_i$ for $1 \leq i \leq k-1$ 
and the (non-collinear) triangles\fabr{This def. is very hard to parse} $w_0q_0w_2$ and $w_tq_tw_{t-2}$ have the same orientation (i.e., are either both positively or both negatively oriented).%
\footnote{Any container is thus \emph{weakly simple} according to the definitions in \cite[Box 5.1]{demaine2007geometric} and \cite{kusakari1999shortest}.
The concept of weakly simple polygons is extensively discussed in \cite{chang2014detecting}.}
For two curves formed by at most $k$ \edi{line segments} in total, it can be decided in time $O(k^3)$ whether there exists a crossing among them or not~\cite{chang2014detecting}.
With this definition, it is guaranteed that every container has a well-defined interior \cite{chang2014detecting}.

\begin{figure}%
    \centering
    \subfloat[A container with $\kappa = 5$. The line segment $s_4$ on the boundary of the container is degenerate. The curves $f_1$ and $f_5$, as well as $f_1$ and $f_2$, respectively, touch on the green segments but do not cross.
    ]
    {\includegraphics[width=0.47\textwidth]
    {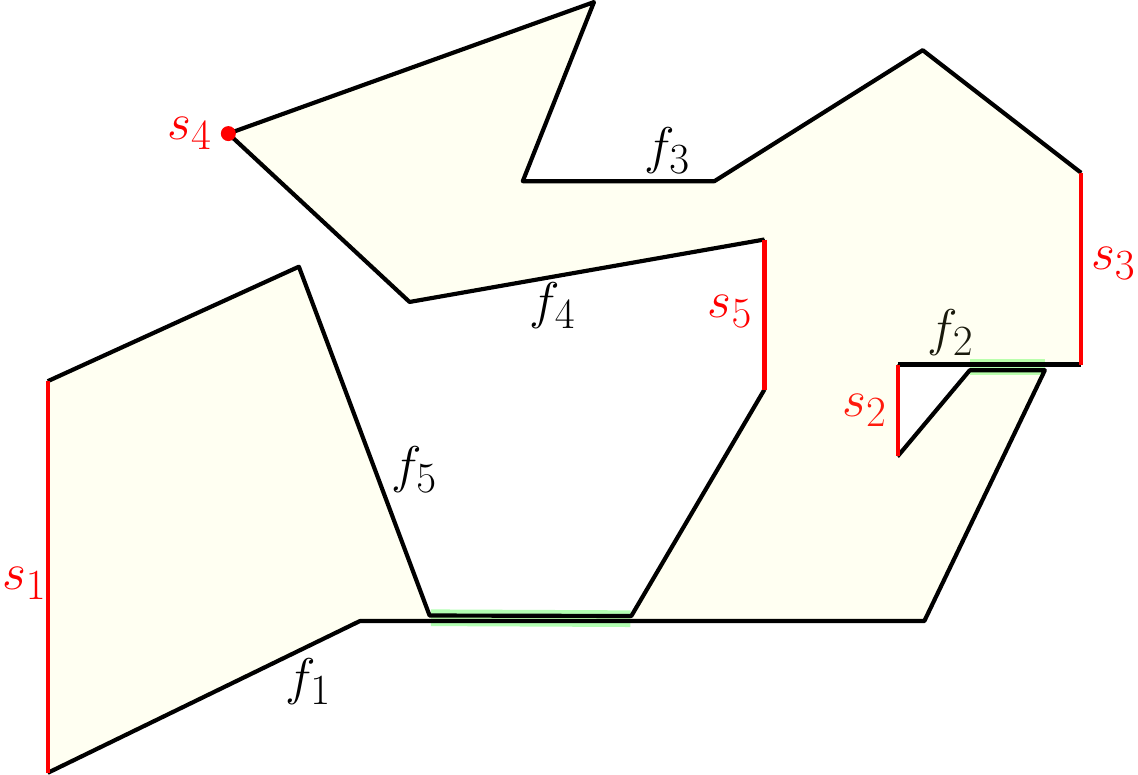} }%
    \qquad
    \subfloat[The curves on the left touch without crossing: the triangles $w_0q_0w_2$ and $w_4q_4w_2$ have negative and positive orientation, respectively.
    The curves on the right cross: the triangles $w'_0q'_0w'_2$ and $w'_4q'_4w'_2$ are both negatively orientated.
    ]{{\includegraphics[width=0.40\textwidth]{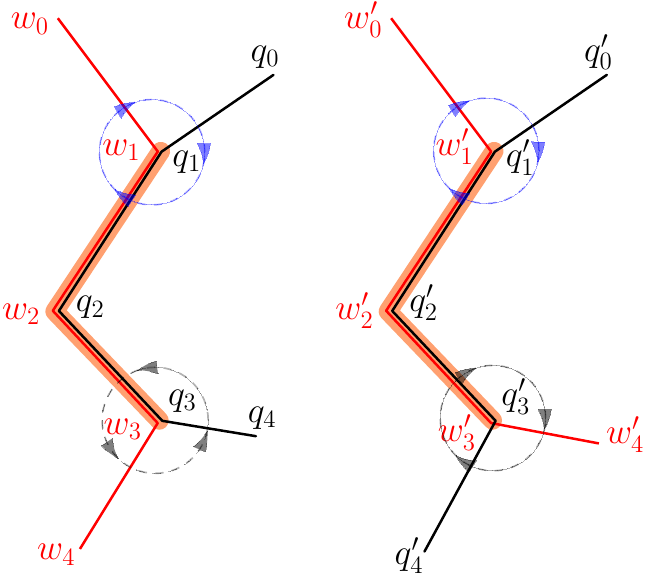} }}%
    \caption{A container with $\kappa = 5$. 
    An illustration of crossing and non-crossing.
    }
    \label{fig:weaklySimple}
\end{figure}

\section{Our Approach}
\label{sec:approach}

First, we present the algorithm in Section~\ref{sec:dynamicProgram}, and give an overview of the analysis in Sections~\ref{sec:overviewAnalysis} and \ref{sec:overviewProposition}. The detailed analysis and proofs are given in the later sections.

\subsection{The algorithm}
\label{sec:dynamicProgram}
Our algorithm is a dynamic program that generalizes the algorithm in~\cite{galvezSODA}. 
Each cell of the dynamic program corresponds to a container $C \in \calc$.\enote{From here on only $\calg_{2d}$ should be used; not $\calg$.} 
For each container, the dynamic program computes a set of disjoint polygons $\Dyn(C) \subseteq \cali$ as follows.
If $C$ encloses no polygon in $\cali$, set $\Dyn(C) = \emptyset$. If $C$ encloses exactly one polygon $P \in \cali$, set $\Dyn(C) = \{P\}$.
Otherwise, the dynamic program goes through all bipartitions of $C$ and chooses the bipartition $\{C_1, C_2\}$ that maximizes $|\Dyn(C_1)| + |\Dyn(C_2)|$ and sets $\Dyn(C) = \Dyn(C_1) \cup \Dyn(C_2)$.
The final output of the algorithm is $\Dyn(C^*)$.

\begin{lemma}[Running time]
\label{lem:runningTime}
    Let $N = |\calv_{2d}|$ be the number of points in the grid $\calg_{2d}$. 
    $\Dyn(C^*)$ can be computed in time $O\big(N^{20d+20}\big) = O((nd)^{O(d4^d)})$.
    \anote{can't we remove the big O in front of $(nd)^{O(d4^d)}$?}
\end{lemma}

\begin{proof}
    The boundary of each container can be identified by a sequence of $10d+10$ line segments in $\calg_{2d}$.
    There are therefore at most $O\big(N^{10d+10}\big)$ containers in $\calc$.
    As argued in \cite{galvezSODA}, any bipartition $\{C_1, C_2\}$ of $\calc$ is determined by the boundary between $C_1$ and $C_2$, i.e., $\partial C_1 \cap \partial C_2$,
    which is composed of at most $10d+10$ line segments.\fabr{Use edge only for sides of polygons and check if edge is defined this way} 
    Thus, to compute $\Dyn(C)$, the dynamic program does not consider more that $O\big(N^{10d+10}\big)$ bipartitions. 
    This gives a total running time $O\big(N^{20d+20}\big)$.
    The lemma follows since $N = O((2d^3n)^{4^d})$, see Section~\ref{sec:preliminaries}.
\end{proof}
\anote{we are a bit sloppy with running time, we only count the number of DP cells (containers) and of branches (bipartitions), but not the ``preprocessing'' of finding/verifying all containers and for each container, find all input polygons inside it. We for sure miss a term $O(N^{\poly d}$.}

It is not hard to see that the output $\Dyn(C^*)$ is indeed an independent set, so we will focus on showing that the algorithm has the claimed approximation guarantee.

\subsection{Analysis}
\label{sec:overviewAnalysis}

By construction, the output solution $\Dyn(C^*)$ is  the union of the solutions of two smaller containers, and so on.
We represent this structure by a binary tree called \emph{recursive partition} defined below. 
We argue that $\Dyn(C^*)$ is the best solution among all the solutions representable by a recursive partition. 
Then, we show the existence of a recursive partition that respects the approximation factor claimed in Theorem~\ref{thm:main}.

\begin{definition}
\label{def:recursivePartition}
    For a set $\calr \subseteq \cali$, a \emph{recursive partition} of $\calr$ is a rooted tree $T$ with vertex set $V$ such that
    \begin{itemize}
        \item every node $u\in V$ corresponds to a pair $(C_u, \pr(C_u))$ where  $C_u  \in \calc$ is a container, and $\pr(C_u)$ is the set of \emph{protected} polygons of $\calr$ contained in $C_u$,
        \item the root $r$ of $T$ corresponds to $(C^*, \emptyset)$, i.e., $C_r = C^*$ and $\pr(C_r) = \emptyset$;
        \item every internal node has two children $u_1, u_2$ such that: $C_{u_1}$ and $C_{u_2}$ form a bipartition of $C_u$, and $\pr(C_{u}) \subseteq \pr(C_{u_1}) \cup \pr(C_{u_2})$;
        \item for every leaf $u$ of $T$, $C_u$ contains exactly one polygon $P_u \in \calr$ or no polygon in $\calr $ at all;
        \item for every $P \in \calr $, there exists a leaf $u$ of $T$ such that $P$ lies in $C_u$.\enote{We need also non-protected polygons.} 
    \end{itemize}
\end{definition}

Clearly, if $\calr \subseteq \cali$ admits a recursive partition, it must be an independent set.
It is  easy to show by induction on the height of the tree that the output $\Dyn(C^*)$ admits a recursive partition, which leads to the following lemma.

\begin{lemma}[{\cite[Lemma~2.2]{galvezSODA}}]
    If $\calr  \subseteq \cali$ admits a recursive partition, then $|\Dyn(C^*)| \geq |\calr |$.
\end{lemma}

Therefore, Theorem~\ref{thm:main} is a consequence of Lemma~\ref{lem:runningTime} and the following proposition.

\begin{restatable}{ourProposition}{mainProp}\label{prop:existenceRecursivePartition}
Let $\OPT$ be an optimal solution of an instance of MISP. There exists a recursive partition for some set $\calr  \subseteq \OPT$ such that $|\calr | \geq \frac{3}{8d}|\OPT|$.
\end{restatable}

\subsection{Informal overview of the proof of Proposition~\ref{prop:existenceRecursivePartition}}
\label{sec:overviewProposition}
Intuitively, we construct the set $\calr $ by starting from an optimal solution $\OPT$ contained in the initial container (the bounding box) $C_{r} = C^*$ and $\pr(C_r) = \emptyset$. 
Then, we will recursively partition the current container $C_u$ into two containers $C_{u_1}$ and $C_{u_2}$.  
$\calr$ is then defined as the set of polygons of $\OPT$ that are fully contained in the leaf containers.
\edi{For a polygon $P$ in $C_u$, we say that $P$ is \emph{lost (at $C_u$)} if it is neither contained in $C_{u_1}$ nor in $C_{u_2}$.}\fabr{Define ``lost'' here. Maybe you would need $\OPT(C_u)$ here to do that} 

Below, one of the $d$ directions in $\cald$ plays a special role: without loss of generality, we assume that this direction is vertical/vertical-up ($v_1$). The exact choice will be made later.

\smallskip \noindent \textbf{Accountable polygons.}
We prove that there exists a subset $\cP\subseteq \OPT$ (the \emph{\edi{accountable}} polygons) with at least $\frac{3}{4d}|\OPT|$ polygons, such that for each polygon $P\in \cP$ lost during partitioning of some $C_u$ into $C_{u_1}$ and $C_{u_2}$
we can \emph{charge} a \edi{unique} polygon $P'\in \OPT$ and $P'$ lies in a leaf container of the recursive partition.

We next describe in more details the set of \edi{accountable polygons} $\cP$ and how \edi{protected} polygons are defined. For technical reasons, we replace each original polygon $P\in \OPT$ with a new polygon $\ext(P)$ lying on $\calg_{2d}$ that contains $P$ (see Figures~\ref{fig:extending1} and~\ref{fig:seeing}). The new set of polygons remains independent, and we will simply denote it by $\OPT$ in the following.

Let $P\in \OPT$ and consider its edge $e_1(P)$ in direction vertical-up.
Let $P'\in \OPT$ and consider its edge $e_{d+1}(P')$ in direction vertical-down. 
We say that $P$ \emph{sees} $P'$ if $e_1(P)$ is non-degenerate and $h(e_{d+1}(P'))\in \interior(e_1(P))\cup \{t(e_1(P))\}$, see Figure~\ref{fig:seeing}.
We let the set $\cP$ of \edi{accountable} polygons be the polygons $P\in \OPT$ such that $P$ sees some $P'\in \OPT$.

\anote{I removed the lemma saying $|\cP|\ge \frac{3}{4d} |\OPT|$, it's redundant and we never refer to it}
It is easy to show that each polygon is seen by at most one other polygon in $\OPT$.

\smallskip
\noindent \textbf{Partitioning.}
For $C \in \calc$, let $\OPT(C)$ be the set of polygons in $\OPT$ that lie on $\interior(C)$.
Our construction is guided  by a partitioning lemma which is stated later. 
Roughly speaking,  
let $C$ be a container with $|\OPT(C)| \ge 2$, and let $\pr(C)$ be the set of \edi{protected} polygons in $C$.
The partitioning lemma states that $C$  
can be bipartitioned by a curve $\Gamma$ into two smaller containers $C_1$ and $C_2$ such that
    \begin{enumerate}[label=(P\arabic*), leftmargin=*] 
        \item\label{verticalLineOVERVIEW} $\Gamma$ contains a vertical line segment $\ell$ that intersects all the polygons in $\OPT(C)$ that are intersected by $\Gamma$.
        \item\label{protectedSafeOVERVIEW} $\Gamma$ does not intersect any polygon in $\pr(C)$,
        \item\label{protectedGrows} $\pr(C) \subseteq \pr(C_1) \cup \pr(C_2)$.
    \end{enumerate}
We stress that the lemma does \emph{not} hold for an arbitrary set $\pr(C)$ (e.g., if we take $\pr(C) = \OPT(C)$).
The set of \emph{protected} polygons in a container is \edi{defined below}.

\smallskip
\noindent \textbf{Charging and protecting.} The recursive partition which determines $\calr$ is defined by repeatedly applying the partitioning lemma. 
During the construction of the recursive partition, we need to guarantee that the vertical line segments given by~\ref{verticalLineOVERVIEW} do not intersect too many polygons from $\OPT$; this is the only possibility of ``losing'' some polygons.
For this, we use the set of \edi{accountable} polygons $\cP \subseteq \OPT$. 
Whenever we apply the partitioning lemma, the line $\ell$ intersects some polygons in $\cP$.
For each $P\in \cP$ that is intersected by $\ell$, \edi{i.e., for each lost polygon $P\in \cP$, we charge exactly one polygon $P'$ seen by $P$. 
By~\ref{verticalLineOVERVIEW}, if $\ell$ intersects $P$, then $\Gamma$ does not intersect $P'$.
If $P'$ is not already an element of $\pr(C)$ and thus an element of $\pr(C_{1}) \cup \pr(C_{2})$, }
then we add the polygon $P'$ 
to either $\pr(C_1)$ if $P'\in C_1$ or to $\pr(C_{2})$ if $P'\in C_2$. 
Moreover, if there is a polygon $P'' \in \OPT(C)$ that sees $P$, then $P''$ is also added to either $\pr(C_1)$ or $\pr(C_2)$.

\edi{By~\ref{protectedGrows}, adding $P'$ to one of $\pr(C_1)$ and $\pr(C_2)$ means that \edi{the charged polygon} $P'$ will remain \edi{protected}. By~\ref{protectedSafeOVERVIEW}, $P'$ will not be intersected by the curves in the following applications of the partitioning lemma.}
Therefore $P'$ will be an element in $\calr$ (our intended recursive partition).
Adding $P''$ to one of $\pr(C_1)$ and $\pr(C_2)$ is also necessary, because the polygon $P$ is already lost and if we were to \edi{lose} $P''$ in one of the following steps, there might not be a polygon which we could \edi{charge} the loss of $P''$ to. 

We conclude that for every polygon $P\in \cP$ lost in the partitioning of a container, we can guarantee that a unique polygon $P'$ seen by $P$ is charged, and it will become the protected polygon in a leaf.
\edi{At least half of the polygons in $\cP$ are either lost or not, so there are at least $\frac{1}{2}|\cP|$ polygons in the leaves.}
\cref{prop:existenceRecursivePartition} follows since $|\cP| \ge \frac{3}{4d} |\OPT|$.

\subsection{Comparison with previous work on MISR}
\label{sec:comparison}
Overall, we follow the same high level approach as the papers on MISR~\cite{galvezSODA,galvez20212+,mitchell2022approximating}.
Yet, to generalize the results on MISR to MISP,  we encounter several technical difficulties. 
We discuss a few of the more prominent ones below.

\edi{To define the set $\cP$, we need the following property (later referred as \ref{lem:extensionCASE4}): for every $P \in \OPT$ and every non-degenerate edge $e$ of $P$, $\interior(e)$ touches either another polygon $P' \in \OPT$ or the boundary of the bounding box.
This property can be obtained by ``maximally extending'' $\OPT$ as in~\cite{galvezSODA,mitchell2022approximating}. 
The difficulty here, unlike in the case of rectangles, is that naively extending the polygons can result in a grid of exponential size in $n$.}

For MISR~\cite{galvezSODA,mitchell2022approximating}, the accountable polygons correspond to the non-nested polygons (both vertical and horizontal). 
It is essentially trivial to show that the number of non-nested rectangles is at least half of the optimal number of rectangles.
In case of convex polygons, we require a more careful argument to show that there are at least $\frac{3}{4d}|\OPT|$ accountable polygons.

To obtain the partitioning lemma, we follow the same idea as in the case of axis-parallel rectangles but we need to work with significantly more complex objects. 
Firstly, the containers we work with have $O(d)$-times more line segments. 
Secondly, the containers that appear in our construction might not be simple (since some parts of the boundary may touch other parts of the boundary).
These difficulties require more elaborate and more technical arguments.


\section{Charging options \edi{and accountable polygons}}
\label{sec:good}
\label{sec:goodCharging}

Like the papers~\cite{galvezSODA,mitchell2022approximating} on MISR, first, we extend an optimum solution $\OPT$. 
\begin{restatable}{ourDefinition}{extension}
\label{def:extension}
Let $\OPT$ be an optimal solution of a MISP instance. 
We say that $\OPT'$ is a \emph{maximal extension} of $\OPT$ if:
\begin{enumerate}[label={(E\arabic*)}, leftmargin=*]
    \item\label{lem:extensionCASE3} $\OPT'$ is an independent set of (convex) polygons on $\calg_{2d}$ and enclosed in $C^*$.
    \item\label{lem:extensionCASE1} There exists a bijection $\ext : \OPT \to \OPT'$ such that $P\subseteq \ext(P)$ for every $P \in \OPT$.
    \item\label{lem:extensionCASE4} For every $P \in \OPT'$ and every non-degenerate edge $e$ of $P$, $\interior(e)$ touches either another polygon $P' \in \OPT'$ or $\partial C^*$.
\end{enumerate}
\end{restatable}

In Appendix~\ref{appendix:extension}, we show that a maximal extension of $\OPT$ exists. 
The idea is to start with $\OPT$ and whenever some non-degenerate edge $e$ of a polygon $P$ does not satisfy~\ref{lem:extensionCASE4}, then we \edi{extend} the polygon $P$ by moving the edge $e$ ``outside'', see Figure~\ref{fig:extending1}.
We show that if we extend all the edges of the same direction not satisfying~\ref{lem:extensionCASE4} together, then the maximal extension lies on the grid $\calg_{2d}$.

By~\ref{lem:extensionCASE1} and~\ref{lem:extensionCASE3}, it suffices to prove \cref{prop:existenceRecursivePartition} for a maximal extension of $\OPT$.  (In particular, \ref{lem:extensionCASE3} implies that the polygons in $\OPT'$ have edges in the given $d$ directions.)
\ant{The purpose of a maximal extension is~\ref{lem:extensionCASE4}, which is helpful to bound the number of accountable polygons.}
For the rest of the paper, we assume that $\OPT$ is already ``maximally extended'' and thus satisfies~\ref{lem:extensionCASE4}, and we work with the grid $\calg_{2d}$.

\begin{figure}
\centering
\begin{minipage}[t]{.35\textwidth}
  \centering
  \includegraphics[width=.89\linewidth]{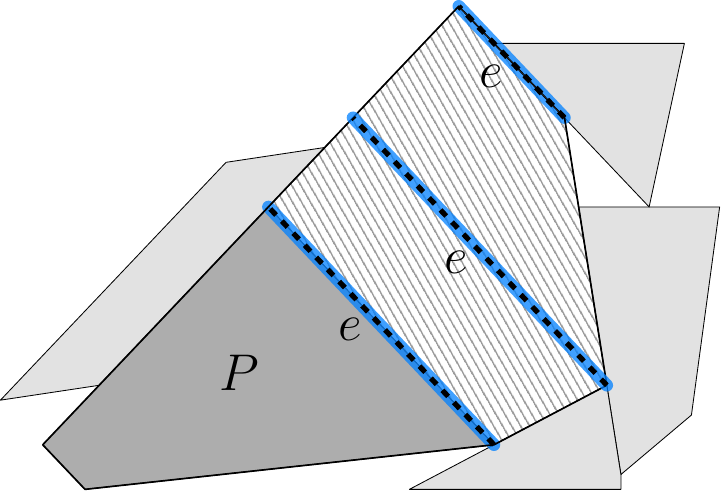}
  \captionof{figure}{\edi{Illustration of the process of extending a polygon $P$. We extend $P$ by moving the edge $e$ of $P$ until $\interior(e)$ touches another polygon in $\OPT$.}}
  \label{fig:extending1}
\end{minipage}%
\hspace{.04\textwidth}
\begin{minipage}[t]{.6\textwidth}
  \centering
  \includegraphics[width=0.86\textwidth]{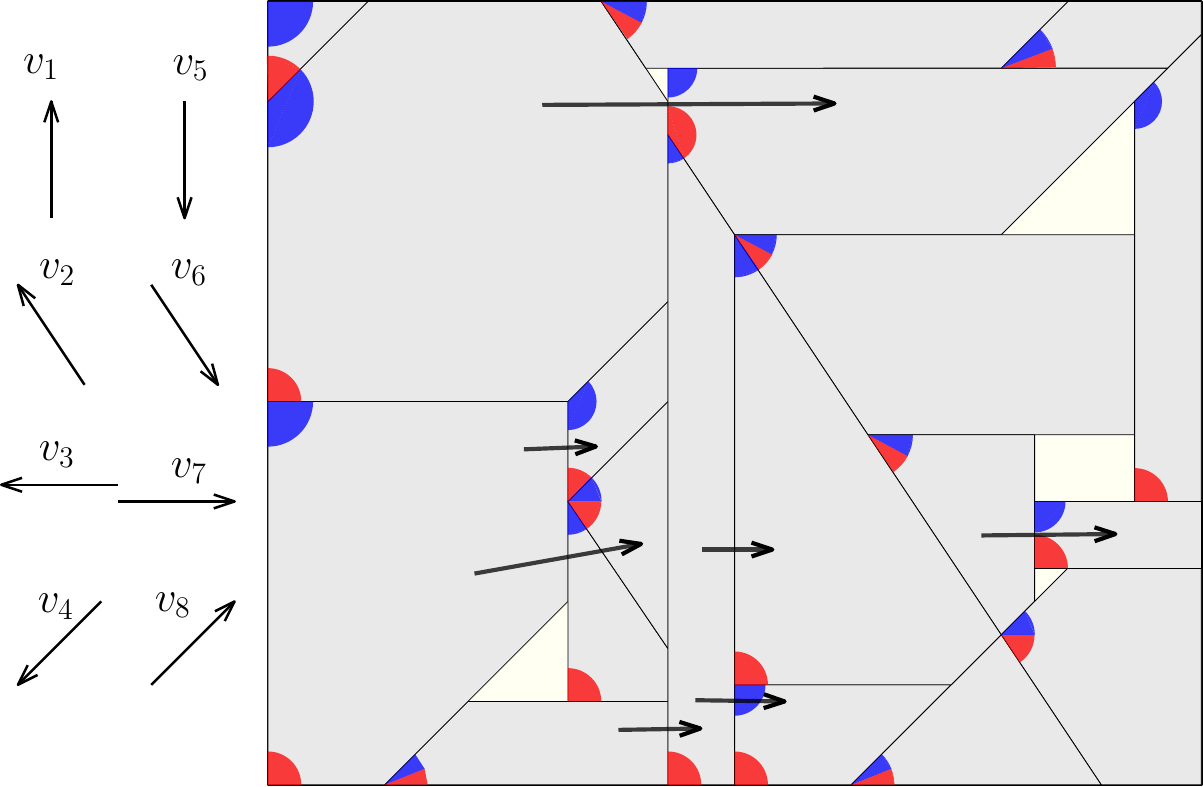}%
  \captionof{figure}{
  A black arrow from $P$ to $P'$ indicates that $P$ sees $P'$ with respect to the option $(v_1, t)$, i.e.,  direction vertical-up and tail.
  The blue (resp. red) corners represent the tails (resp. head) of all edges with direction vertical-down ($v_5$). Thus, a polygon $P$ sees a polygon $P'$ if the vertical-up edge of $P$ is touching the red corner of $P'$.
  }%
    \label{fig:seeing}%
\end{minipage}%
\end{figure}

\medskip
In the rest of this section, by the term \emph{direction} we mean a direction $v_i$ where $i\in [2d]$, and say that edge $e$ is of direction $v_i$
if the points of the edge $e$ correspond to $t(e)+\lambda\cdot v_i$, with $\lambda \ge 0$.  
A \emph{charging option} is specified by a direction $v_i$, $i\in[2d]$ and a choice between $t$ and $h$.
Let $\calo=\{v_i\}_{i \in [2d]} \times \{t,h\}$ be the set of the $2d\cdot 2=4d$ charging options.
We show the existence of a charging option and a subset $\cP \subseteq \OPT$ of accountable polygons with respect to this option such that (essentially) $|\cP|\ge \frac{3}{4d} |\OPT|$.

\begin{definition}
\label{def:chargable}
    Let $P\in \OPT$ and let $e$ be the edge of $P$ in direction $v=v_i$, $i\in [2d]$.
    \begin{itemize}
        \item Let $P'\in \OPT$ and $e'$ be the (possibly degenerate) edge of $P'$ of direction $-v$. 
        For $a \in \{t,h\}$,  we say that $P$ \emph{sees $P'$ with respect to  $(v, a)$} if $e$ is non-degenerate and if $\neg a(e') \in \interior(e)\cup \{a(e)\}$,
        where $\neg t = h$ and $\neg h = t$. (See Figure~\ref{fig:seeing}.)
        \item Whenever there exists $P'\in \OPT$ and a charging option $(v,a)$, such that $P$ sees $P'$ for $(v,a)$ then we say that $P$ is \emph{accountable for $(v, a)$}.
    \end{itemize}
\end{definition}

\begin{lemma}
\label{lem:chargingIsInjective}
    Let $(v, a)\in \calo$ be a charging option. 
    Any polygon $P'\in \OPT$ is seen by at most one other polygon $P\in \OPT$ with respect to $(v, a)$. 
\end{lemma}
\begin{proof}
    Assume that $P'$ is seen by $P_1, P_2 \in \OPT$ with respect to $(v,a)$.
    Let $e_1$ and $e_2$ be the edge in direction $v$ of $P_1$ and $P_2$, respectively.
    Then we have $\neg a(e') \in (\interior(e_1) \cup \{a(e_1)\}) \cap (\interior(e_2) \cup \{a(e_2)\})$. Since $\interior(e_1) \neq \emptyset$ and $\interior(e_2) \neq \emptyset$, it follows that $\interior(e_1) \cap \interior(e_2) \neq \emptyset$.
    This implies that $P_1$ and $P_2$ intersect, thus $P_1 = P_2$.
\end{proof}

We say that a polygon $P\in \OPT$ is a \emph{corner polygon} in the bounding box $C^*$, if all but one of the edges of $P$ are contained in the boundary of $C^*$.
In particular, $P$ is a corner polygon if $P = C^*$. Similarly, if $C^*$ is partitioned into two convex polygons, then both are corner polygons.
Let $Z \subseteq \OPT$ be the set of corner polygons in $C^*$. Since $C^*$ is a parallelogram, we have $|Z|\le 4$, and the polygon $C' = C^* \setminus (\bigcup Z)$ is convex.\fabr{I can't see how $|Z|>3$ is possible \edi{Take a square of side $s$. Place a \emph{corner triangle} (with sides less than $s/2$) in each corner?}}

\begin{lemma}[Good charging option]
\label{lem:exists_good_option}
    Assume that $\OPT$ satisfies~\ref{lem:extensionCASE4}.
    Then, there exists a charging option $(v, a)\in \calo$ such that at least $\frac{3}{4d}|\OPT\setminus Z|$ polygons in $\OPT\setminus Z$ are accountable with respect to $(v,a)$. 
\end{lemma}

\begin{proof}
    Let $P \in \OPT$ and $c$ be a vertex of $P$. 
    Let $e, e'$ be the two non-degenerate edges incident to $c$ where $c=h(e)=t(e')$. Denote with $v$ (resp. $v'$) the direction of $e$ (resp. $e'$).
    \begin{claim}\label{claim:twoConsecutiveSides}
        Suppose that $e$ or $e'$ (or both) does not lie on the boundary of $C^*$.
        Then, $P$ is accountable with respect to $(v,h)$ or $(v',t)$.
    \end{claim}
    \begin{claimproof}
    By~\ref{lem:extensionCASE4}, each non-degenerate edge of $P$ not contained in the boundary of the bounding box, must touch some other polygon of $\OPT$ in its interior. 
    By assumption either $e$ or $e'$ does not lie on the boundary of $C^*$, without loss of generality, say $e$. 
    Then $P$ touches some $P_1\in \OPT$ on $\interior(e)$, i.e., $\interior(e) \cap e_1 \neq \emptyset$, where $e_1$ is the edge of $P_1$ in direction $-v$ ($e_1$ could be degenerate). See Figure~\ref{fig:claimTwoConsecutiveSide}.
    If $P$ sees $P_1$ with respect to $(v,h)$, i.e.,  $t(e_1) \in \interior(e) \cup \{h(e)\}$ then the claim is true, so assume that $t(e_1) \notin \interior(e) \cup \{h(e)\}$.
    This however implies $c \in \interior(e_1)$.

    Since $c \in \interior(e_1)$ and $C^*$ is convex, it follows that $e'$ is not on the boundary of $C^*$.
    Then, by \ref{lem:extensionCASE4}, there exists $P_2 \in \OPT$ that touches $P$ on $\interior(e')$, i.e., $\interior(e') \cap e_2 \neq \emptyset$, where $e_2$ is the edge of $P_2$ in direction $-v'$. 
    If $P$ does not see $P_2$ with respect to $(v',t)$, then $c \in \interior(e_2)$ by the same argument as before. 
    So $\interior(e_1)$ and $\interior(e_2)$ intersect in $c$ and thus $P_1$ and $P_2$ intersect (as $e_1$ and $e_2$ have different direction) which is a contradiction. 
    Therefore, $P$ must see $P_2$ with respect to $(v',t)$.
    \end{claimproof}
    \begin{figure}%
    \centering
    \includegraphics[width=0.7\textwidth]{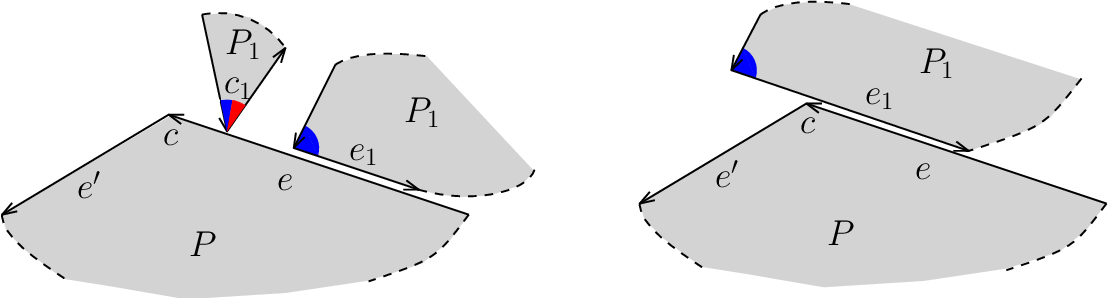}%
    \caption{Claim~\ref{claim:twoConsecutiveSides}: the blue (red) corners represents the tail (head) of the edges in direction $-v$.}%
    \label{fig:claimTwoConsecutiveSide}%
    \end{figure}

    Consider $P\in \OPT\setminus Z$.
    Since $P$ is not a corner polygon in $C^*$, it has at least two consecutive non-degenerate edges such that neither of them lies on $\partial C^*$.
    By Claim~\ref{claim:twoConsecutiveSides}, every vertex of $P$ incident to one or both of these edges, provides a charging option for which $P$ is accountable.
    Thus, the total number of pairs $(P, (v, a))$ with $P\in\OPT \setminus Z$ and $(v, a)\in\calo$ such that $P$ is accountable with respect to $(v, a)$ is at least $3|\OPT \setminus Z|$.
    Since $|\calo|=4d$, there exists an option $(v, a)$ for which the number of accountable polygons in $\OPT\setminus Z$ is at least $\frac{3}{4d}|\OPT\setminus Z|$.%
    \footnote{If we could guarantee a maximal extension in which all the polygons have at least $4$ sides, then we would improve $\frac{3}{4d}$ to $\frac{1}{d}$. In particular, when $d=2$ we are in the case of axis-parallel rectangles and we obtain a $2d = 4$-approximation algorithm. This is the same approximation factor achieved in \cite{galvezSODA,galvez20212+,mitchell2022approximating}
    by charging each lost rectangle to one protected rectangle (the improved $2+\epsilon$ factor requires a more complex charging).}
\end{proof}


\section{Recursive partitioning}
\label{sec:recursivePartitioning}

\fabr{moved this paragraph here from the end of the previous section} 
Without loss of generality (by rotating and mirroring the initial instance if necessary), we assume that the option $(v, a)$ satisfying Lemma \ref{lem:exists_good_option} is  vertical-up and tail, i.e., $(v_1, t)$.
In other words, for any $P \in \OPT$, if $e_1(P)$ is non-degenerate and if there is a $P' \in \OPT$ such that $h(e_{d+1}(P')) \in \interior(e_1(P)) \cup t(e_1(P))$, then we say that \emph{$P$ sees $P'$} (and \emph{$P'$ is seen by $P$}) and that $P$ is \emph{accountable}.
Lemma~\ref{lem:exists_good_option} states that there exists a subset $\cP \subseteq \OPT \setminus Z$ of accountable polygons such that $|\cP| \ge \frac{3}{4d}|\OPT\setminus Z|$, consequently $|Z| + |\cP| \ge \frac{3}{4d} |\OPT|$.

We will construct a recursive partition for a specific subset $\calr \subseteq \OPT$,
such that $|\calr| \ge|Z| + \frac{1}{2} |\cP|$, which proves \cref{prop:existenceRecursivePartition}.
Recall that $\OPT(C)$ denotes the set of polygons in $\OPT$ that lie on $\interior(C)$.
Moreover, all of the polygons in $\OPT$ and the bounding box $C^*$ lie on the grid $\calg_{2d}$.

\medskip
\noindent\textbf{Handling corner polygons.}
If $Z\neq \emptyset$, then we construct the first few nodes of the recursive partition as follows. 
Take any corner polygon $P\in Z$.
Recall that the root $r$ of the recursive partition corresponds to $(C^*, \emptyset)$.
We add two children $u_1, u_2$ to $r$ and partition $C^*$ into the containers $C_{u_1} = P$ and $C_{u_2} = C^*\setminus P$. Set $\pr(C_{u_1}), \pr(C_{u_2}) = \emptyset$.
By construction, $\OPT(C_{u_1}) = \{P\}$ 
(so $u_1$ is a leaf in the final tree and $\OPT(C_{u_2}) = \OPT \setminus \{P\}$. 
Notice that $C^*\setminus P$ is convex with at most five line segments since $C^*$ is convex. $C^*\setminus P$ has five line segments if $P$ is a triangle, and less if $P$ has more than three sides.)
We recurse by treating $C_{u_2}$ as the new bounding box.

We end up with a tree on $|Z| + 1$ leaves, where for one leaf $u$, $C_u$ is a convex polygon such that $\OPT(C_u) = \OPT \setminus Z$ and with at most eight line segments (since $|Z| \leq 4$) and $\pr(C_u) = \emptyset$.
Each of the remaining $|Z|$ leaves coincides with a unique element in $Z$.
Thus, it suffices to construct the recursive partition of $\OPT \setminus Z$ by treating $C_u$ as the bounding box with at most $8$ line segments.
Equivalently, we assume $Z = \emptyset$ and allow $C^*$ to have up to eight line segments for the rest of this paper.

\subsection{The partitioning lemma -- formal statement}
\label{sec:partitioning1}

For any $P \in \OPT$, let the \emph{top of} $P$ be defined as the curve $\ttop(P) = e_{2}(P) e_3(P) \cdots e_{d}(P)$ and the \emph{bottom of $P$} as the curve $\bbot(P) = e_{d+2}(P) e_{d+3}(P) \cdots e_{2d}(P)$. 
We define the bottom and top of the bounding box $C^*$ in the same way.
The following definitions are illustrated in Figure~\ref{fig:fenceProtected}.

\begin{definition}[Top and bottom fences]
\label{def:fence}
\begin{sloppypar}
    Let $P, P' \in \OPT$ be two polygons such that $P$ sees $P'$. 
    A \emph{top-fence} is (a segment of)
    the curve $\ttop(P) \overline{h(e_1(P)) t(e_{d+1}(P'))} \ttop(P')$ 
    such that the first and last line segment is not vertical.
    Symmetrically, a \emph{bottom-fence} is (a segment of)
    the curve $\bbot(P) \overline{t(e_1(P)) h(e_{d+1}(P'))} \bbot(P')$
    such that the first and last line segment is not vertical.

    If $P \in \OPT$ does not see any polygon, then a segment of its bottom (or top) is also called a bottom-fence (resp. top-fence).
    
    For a vertical line segment (cutting line) $s$, we say that a fence \emph{emerges from $s$} if one extreme point of the fence lies on $s$.
\end{sloppypar}
\end{definition}

To prove the partitioning lemma, we further specialize the definition of a container (see Section~\ref{sec:preliminaries})

\begin{definition}[Structured container]
   \label{def:structured_container}
    A container $C$ with $\partial C = s_1 f_1 s_2 f_2 \cdots s_\kappa f_\kappa$, $\kappa \leq 5$, is \emph{structured} if the \emph{cutting lines} $s_1, \dots, s_\kappa$ are vertical and the curves $f_1, \dots, f_\kappa$ are fences.

    We say that a cutting line is a \emph{left cutting line} if it is oriented \ant{downwards (or degenerate)}, and \emph{right cutting line} \ant{if it is oriented upwards (or degenerate)}. 
    In a structured container, the left cutting lines (and thus right cutting lines) are consecutive (e.g., $s_1, \dots, s_{\kappa'}$ are left and $s_{\kappa'+1}, \dots, s_\kappa$ are right cutting lines for some $\kappa' \in [\kappa-1]$).
\end{definition}

\begin{definition}[Protected \edi{by fences}]
\label{def:protected_polygon}
Let $C$ be a structured container and $s$ be a (possibly degenerate) cutting line on $C$.
We say that a polygon $P \in \OPT(C)$ is \emph{protected from the left in $C$ via $s$} if $s$ is a left cutting line on $\partial C$ and
\begin{itemize}
    \item there exists a top-fence $\gamma_{h}$ in $C$ emerging from $s$, ending in $h(e_1(P))$, and with $\ttop(P) \subseteq \gamma_h$, and
    \item there exists a bottom-fence $\gamma_{t}$ in $C$ emerging from $s$, ending in $t(e_1(P))$, and with $\bbot(P) \subseteq \gamma_t$.
\end{itemize}
\edi{We say that $P$ is protected by fences \emph{$\gamma_h$ and $\gamma_t$}.}
Symmetrically, we say that a polygon $P \in \OPT(C)$ is \emph{protected from the right in $C$ via $s$} if $s$ is a right cutting line on $\partial C$ and
\begin{itemize}
    \item there exists a top-fence $\sigma_{h}$ in $C$ emerging from $s$, ending in $t(e_{d+1}(P))$, and with $\ttop(P) \subseteq \sigma_h$, and 
    \item there exists a bottom-fence $\sigma_{t}$ in $C$ emerging from $s$, ending in $h(e_{d+1}(P))$, and with $\bbot(P) \subseteq \sigma_t$.
\end{itemize}
\edi{We say that \emph{$P$ is protected by fences $\sigma_h$ and $\sigma_t$.}}
\edi{A polygon $P\in \OPT(C)$ is \emph{protected by fences in $C$} if it is either protected from the left in $C$ or protected from the right in $C$.
}
\end{definition}\enote{Any ideas how to improve formatting of this definition?}

\begin{figure}%
    \centering
    \includegraphics[width=0.5\textwidth]{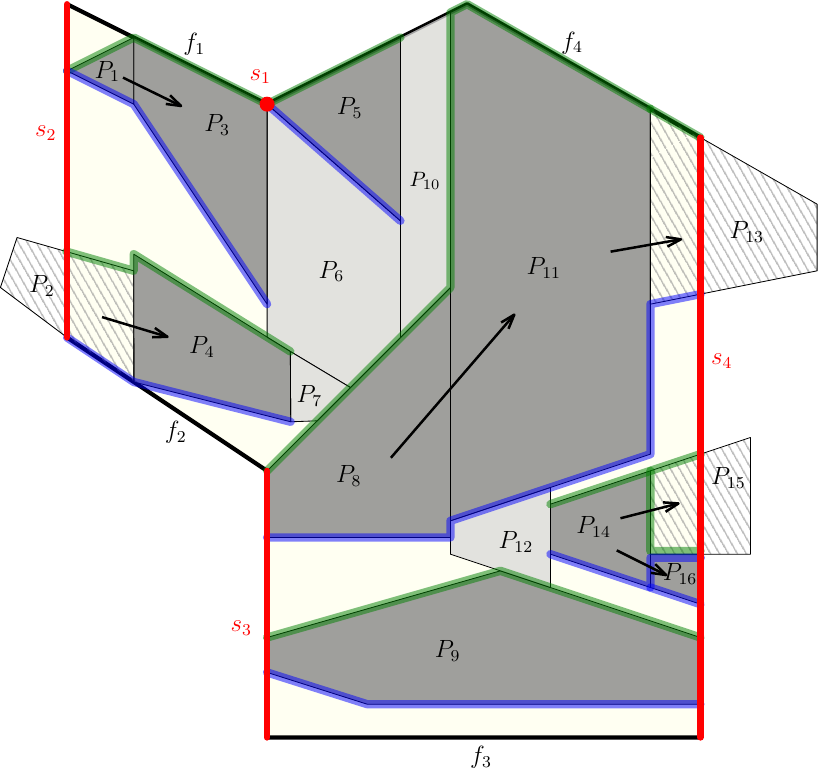}%
    \caption{Example of a structured container with $\kappa = 4$.
    The black arrows represent ``seeing'', top-fences are green, bottom-fences are blue. The polygons $P_1, P_3, P_4, P_5, P_8$ are protected (only) from the left, $P_{14}, P_{16}$ are protected (only) from the right, $P_9, P_{11}$ are protected both from the left and from the right.
    Notice that the fences that protect $P_{14}$ (from the right) are not unique since $P_{14}$ sees $P_{15}$ and $P_{16}$ which are cut and touch $s_4$, respectively.
    Note also that the bottom-fences touching $P_8, P_{11}$ and $P_{11}, P_{13}$ overlap.
    \anote{page 11: In figure 6, it is hard to distinguish the orange segments}
    }%
    \label{fig:fenceProtected}%
\end{figure}

\edi{We will show that each polygon in $\pr(C)$ appearing in the construction of the recursive partition can be protected \emph{by fences} in $C$.}

Next, we state the partitioning lemma, the proof is presented in Appendix~\ref{sec:paritioning2}.
The lemma holds only for \emph{structured} containers.
This matters for the construction of the recursive partition but it does not affect the algorithm, as it  considers all possible containers. 

\begin{restatable}[Partitioning lemma]{ourLemma}{partitioningTrue}
\label{lem:partitioningStronger}
    Let $C$ be a structured container such that $|\OPT(C)|\ge 2$, \edi{and let $\cal P$ be a set of polygons in $C$ protected by fences.}
    Then, there exists a curve $\Gamma$ such that
    \begin{enumerate}[label=(P\arabic*), leftmargin=*] 
        \item\label{strongB} $\Gamma$ partitions $C$ into two structured containers $C_1, C_2 \in \calc$ with non-empty interiors.
        \item\label{strongC} All the polygons in $\OPT(C)$ that are intersected by $\Gamma$ are intersected by one vertical cutting line $\ell\subseteq \Gamma$.
        \item\label{strongD} $\Gamma$ does not intersect any polygon protected by fences.
        \item\label{stronE} \edi{Any polygon protected by fences in $C$ is protected by fences in either $C_1$ or $C_2$.}
    \end{enumerate}    
\end{restatable}

\anote{we might have needed to say that (P5) The boundaries of $C_1, C_2$ overlap on $\Gamma$, (i.e., $\partial C_1 \cap \partial C_2 \subseteq \Gamma$), and $\ell$ forms a cutting line on $\partial C_1$ and $\partial C_2$.}


\subsection{Construction and analysis of the recursive partition}
\label{sec:construction}

In this section we prove~\cref{prop:existenceRecursivePartition}, i.e., we \edi{provide a recursive partition for $\calr \subseteq \OPT$ with }$|\calr|\ge \frac{1}{2} |\cP|$.
(Recall that we already argued that we can assume $Z = \emptyset$.)
We give an iterative construction of a recursive partition with the help of the partitioning lemma.
\enote{The rest of this change quite a bit:}

We initialize a tree $T$ with root node $r$, $C_r = C^*$,  and $\pr(C_r) = \emptyset$. 
Then, iteratively, for every childless node $u \in V(T)$ with $|\OPT(C_u)| \geq 2$, add two children $u_1, u_2$ to $u$ and choose $C_{u_1}, C_{u_2} \in \calc$ as provided by~\ref{strongB} in the partitioning lemma applied to $C_u$ and $\pr(C_u)$.
Define the set of protected polygons $\pr(C_{u_1})$ and $\pr(C_{u_1})$ as follows.
\begin{enumerate}[label=(A\arabic*), leftmargin=*] 
    \item\label{a1} Set $\pr(C_{u_1}) = \pr(C_u) \cap \OPT(C_{u_1})$ and $\pr(C_{u_1}) = \pr(C_u) \cap \OPT(C_{u_1})$.
    \anote{$\pr(C_{u}) \subseteq \pr(C_{u_1}) \cup \pr(C_{u_2})$ follows from Lemma~\ref{lem:fencesFromLeft}, not immediately from the PL}
    \item\label{a2} For each $P\in \cP$ that is intersected by $\ell$, i.e., each $P\in \cP$ that is lost, if $P$ sees a polygon $P' \in \OPT(C_u)$ (if $P$ sees more than one polygon in $\OPT(C_u)$, choose one of them arbitrarily), add $P'$
    to $\pr(C_{u_1})$ if $P'$ is in $C_{u_1}$ or to $\pr(C_{u_2})$ if $P'$ is in $C_{u_2}$. Moreover, charge the loss of $P$ to $P'$.
    \item\label{a3} For each $Q'\in \OPT(C_u)$ intersected by $\ell$ for which there is a polygon $Q \in \OPT(C_u)$ that sees $Q'$, add $Q$ to either $\pr(C_{u_1})$ or $\pr(C_{u_2})$ depending whether $Q$ is in $C_{u_1}$ or $C_{u_2}$.
\end{enumerate}

We first show to that by this construction, a polygon is protected only if it is protected by fences.

\begin{lemma}\label{lem:fencesFromLeft}
    Let $P'\in \pr(C_u)$ for a node $u$ of $T$.
    There exist fences that protect $P'$ in $C_u$.
\end{lemma}

\begin{proof}
    We first argue in the case that $P'$ is protected for the first time, i.e., added to $\pr(C_u)$ via \ref{a2} or \ref{a3}.
    Let $u'$ be the parent of $u$ in $T$.
    
    First assume that $P'$ is protected via \ref{a2}.  
    Let $P\in \cP \cap \OPT(C_{u'})$ be the polygon that sees $P'$.
    By definition, $P$ is intersected by the cutting line $\ell_{u'}$ from \ref{strongB} during the bipartitioning of $C_{u'}$
    Let $p_x$ and $p_y$ be the two intersection points of $\ell_{u'}$ and $\partial P$, where $p_x$ is above $p_y$, see Figure~\ref{fig:lemma23}.
    Since $P$ sees $P'$, the curve $\gamma_x$ on $\ttop(P)$ and $\ttop(P')$ from $p_x$ to $h(e_{1}(P'))$ is a top-fence and the curve $\gamma_y$ on $\bbot(P)$ and $\bbot(P')$ from $p_y$ to $t(e_{1}(P'))$ is a bottom-fence. $\gamma_x$ and $\gamma_y$ both emerge from $\ell_{u'}$ and thus protect $P'$ from the left in $C_{u}$. Hence, $P'$ is protected by fences in $C_{u}$.

    The argument is symmetric if $P'$ is protected via \ref{a3}:
    there is a polygon $Q \in \OPT(C_{u'})$ seen by $P'$ that is intersected by the the cutting line $\ell_{u'}$.
    Therefore, the curves on $\ttop(P')$ and $\ttop(Q)$ from $e_{d+1}(P')$ to $\ell_{u'}$ and of $\bbot(P')$ and $\bbot(Q)$ from $e_{d+1}(P')$ to $\ell_{u'}$ form a pair of fences that protect $P'$ from the right in $C_u$.

    If $P'$ is protected via \ref{a1}, then it has been protected for the first time in an ancestor of $u$, so the claim follows inductively from by \ref{strongD} and \ref{stronE}.
\end{proof}

With \ref{strongD} and \ref{stronE}, Lemma~\ref{lem:fencesFromLeft} implies that protected polygons are not lost and stay protected, i.e.,  $\pr(C_{u}) \subseteq \pr(C_{u_1}) \cup \pr(C_{u_2})$ for every interior node $u$ in $T$. This in particular holds for every charged polygon.
By the construction above, every charged polygon is protected and charged only once by Lemma~\ref{lem:chargingIsInjective}.
To make our charging scheme work, we need to make sure that \emph{every} lost accountable polygon provides one charge, which follows by \ref{strongC} and the following lemma.

\begin{lemma}\label{lem:fencesFromRight}
    Let $P \in \cP$ be a polygon that is intersected by the vertical line segment $\ell_u$ for an internal node $u\in T$.
    Then there exists a polygon $P'\in \OPT(C_u)$ that is seen by $P$.
\end{lemma}
\begin{proof}
    Let $\calp$ be the set of polygons seen by $P$.
    For the sake of contradiction, suppose that $\calp \cap \OPT(C_u) = \emptyset$.
    If some $P' \in \calp$ partially lies in $C_u$, i.e., $P' \cap \interior(C_u) \neq \emptyset$, then $P'$ was intersected by the vertical line $\ell_{u'}$ in an ancestor $u'$ of $u$, so $P$ is protected via \ref{a3}.
    Otherwise, if all polygons in $\calp$ lie outside of $C_u$, then $e_1(P)$ lies on a cutting line in $\partial C_u$.
    Therefore, $\ttop(P)$ and $\bbot(P)$ form a top-fence and a bottom-fence, respectively, that protect $P$ by fences in $C_u$.
\end{proof}
\begin{figure}%
\centering
\includegraphics[width=0.4\textwidth]{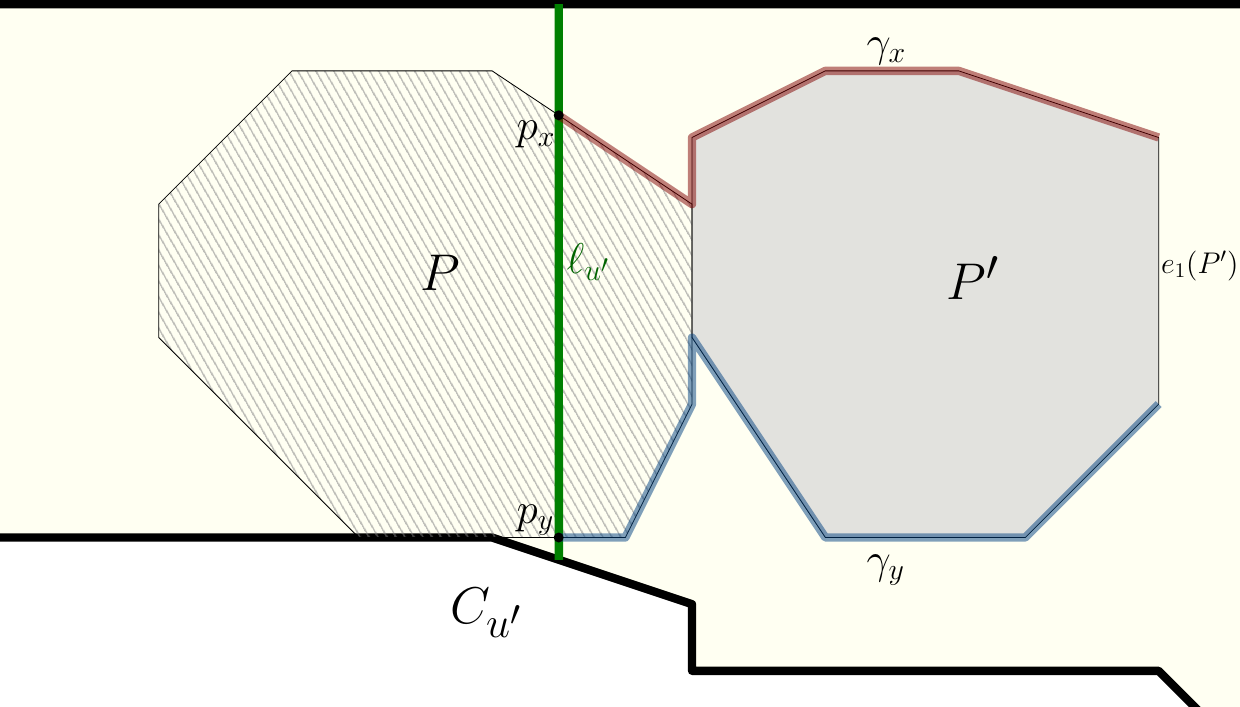}%
\caption{Illustration for the proof of Lemma~\ref{lem:fencesFromLeft}: $P'$ is protected by fences via \ref{a2}.}%
\label{fig:lemma23}%
\end{figure}


\begin{proof}[Proof of Proposition~\ref{prop:existenceRecursivePartition}]
By Lemma~\ref{lem:exists_good_option}, 
we have $|\cP|-|Z| \ge \frac{3}{4d} |\OPT| - |Z|$.
Recall that we have already assigned each polygon of $Z$ to a unique leaf of $T$.
By the charging scheme described above and since a protected (and thus charged) polygon is never lost, we have a unique polygon contained in a leaf of $T$ for each lost accountable polygon during the partition. 
The proposition follows since at least half of the polygons in $\cP$ are either lost, 
or at least half of the polygons in $\cP$ are not lost. 
\end{proof}


\bibliographystyle{plain} 
\bibliography{misr_bib}

\appendix


\section{Maximal extension}
\label{appendix:extension}
For convenience, we restate the definition of a maximal extension. 
\extension*

Recall that $\call_1$ is the set of all lines in directions $v_1, \dots, v_d$ passing through the corners of input polygons, and that $\calv_k$, for $k \in [2d]$ and $\call_k$, for $k\in \{2, \dots, 2d\}$ are recursively defined as follows:
$\calv_k$ is the set of intersection points of any two (non-parallel) lines in $\call_k$,
and $\call_k$ is the set of all lines in directions $\cald$ spanned from all points in {$\calv_{k-1}$}.
In particular, $C^* \in \calg_1$ and $P \in \calg_1$ for every $P \in \OPT$.

\begin{lemma}
There exists a maximal extension of $\OPT$.
\end{lemma}
\begin{proof}
    We describe a process by which, sequentially, every edge of every polygon in $\OPT$ is extended at most once, such that the resulting instance is a maximal extension of the input instance (so we have \ref{lem:extensionCASE1} by construction).
    We use the following claim which is simple to prove.

    \begin{claim}\label{claim:criterion-intersection}
        Two polygons $P, P' \in \cali$ intersect if and only if $p_i(P) > -p_{i+d}(P')$ for every $i \in [2d]$.
    \end{claim}
    
    First, for every $P \in \OPT$, we initialize $D(P)$ as the set of all indices $i\in [2d]$ such that $e_i(P)$ is degenerate.
    We want to assure that for every $i \in D(P)$, $e_i(P)$ stays degenerate during the extension, unless it separates $P$ from another polygon, so we set $p_i(P) = + \infty$ for every $i \in D(P)$.

    We repeat the following sweep line algorithm for every $i \in [2d]$:
    for every polygon $P \in \OPT$ at a time, unless $i \in D(P)$, continuously increase $p_i(P)$ (we say that we \emph{extend $P$ in the $i$-th direction}) until one of the following \emph{events} occurs, see Figure~\ref{fig:maximalExtensionProcess}.
    \begin{enumerate}[label=\textbf{[S\arabic*]}, leftmargin=*]
        \item \label{case:touchPolygon} $P$ touches some $\tilde P \in \OPT$ (or $\partial C^*$) but $\interior(e_i(P))$ does not overlap with $\partial \tilde P$ (or $\partial C^*$) and $P$ intersects $\tilde P$ (or the exterior of $C^*$) if $p_i(P)$ is increased any further.
        \item \label{case:reachPolygon} $\interior(e_i(P))$ overlaps with $\partial \hat P$ for some $\hat P \in \OPT$ or $\partial C^*$.%
        %
        \item \label{case:degenerate} $e_i(P)$ becomes degenerate.
    \end{enumerate}
    To shorten the analysis, we assume without loss of generality that at \ref{case:touchPolygon} and \ref{case:reachPolygon}, we always touch some polygon $\tilde P$ or $\hat P$, respectively, instead of $\partial C^*$.
    Whenever we encounter \ref{case:touchPolygon} (assuming it does not happen at the same time as \ref{case:reachPolygon} or \ref{case:degenerate}), then it means
    that $P$ touches a (non-degenerate) edge $e_{j+d}(\tilde P)$ such that $j \notin \{i, i+d\}$ (and such that $\interior(e_{j+d}(\tilde P))$ intersects $P$ if we extend it further).
    Set $p_{j}(P) = -p_{j+d}(\tilde P)$, remove $j$ from $D(P)$
    and then keep on increasing $p_i(P)$.
    When \ref{case:reachPolygon} occurs, stop increasing $p_i(P)$.
    In case \ref{case:degenerate} occurs, stop increasing $p_i(P)$, add $i$ to $D(P)$ and set $p_i(P) = +\infty$.

    \begin{figure}%
    \centering
    \subfloat[\centering The extension finishes with \ref{case:degenerate} after \ref{case:touchPolygon}.]
    {\includegraphics[width=0.42\textwidth, page=1]{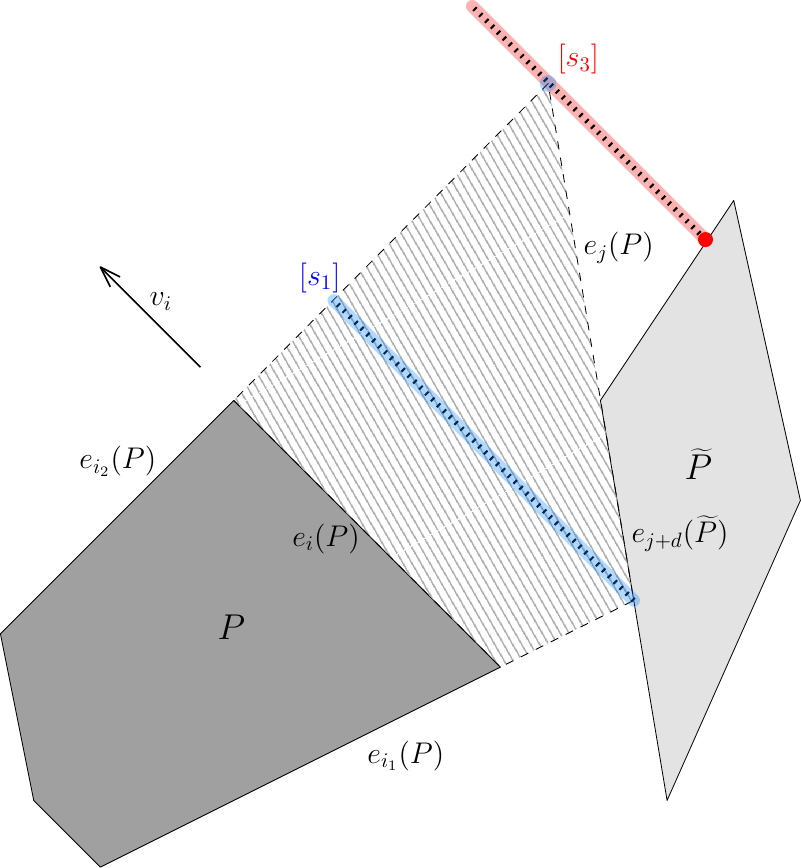} }%
    \qquad
    \subfloat[\centering The extension finishes with \ref{case:reachPolygon}.]{{\includegraphics[width=0.38\textwidth, page=2]{Figures_PDF/fig_maximalExtensionProcess.pdf} }}%
    \caption{
    The polygon $P$ is extended in the $i$-th direction. The position of $e_i(P)$ at each event is highlighted. The point and line highlighted in red lie on $\calg_i$, all other points and lines lie on $\calg_{i-1}$.
    }
    \label{fig:maximalExtensionProcess}
    \end{figure}
    
    It is clear that the extension of $e_i(P)$ must finish by either \ref{case:reachPolygon} or \ref{case:degenerate}, since otherwise $P$ grows indefinitely which contradicts $P \subseteq \interior(C^*)$.
    On the other hand, $P$ can be extended as long as \ref{case:reachPolygon} or \ref{case:degenerate} has not occurred.
    Whenever \ref{case:touchPolygon} occurs, then at the moment of contact, the edges $e_{j+d}(\tilde P)$ and $e_{j}(P)$ overlap, see Figure~\ref{fig:maximalExtensionProcess}(a).
    However since increasing $P$ further in the $i$-th direction makes $P$ and $\tilde P$ intersect, so we have $p_{j}(P) > -p_{j+d}(\tilde P)$ by Claim~\ref{claim:criterion-intersection}.
    This means that the $j$-th constraint of $P$ is redundant and in particular that $e_{j}(P)$ is degenerate until the moment of contact between $P$ and $\tilde P$.
    By setting $p_{j}(P) = -p_{j+d}(\tilde P)$, we can continue to extend $P$ in the $i$-th direction without causing an intersection between $P$ and $\tilde P$.
    Notice that starting from now, increasing $p_j(P)$ causes an intersection between $P$ and $\tilde P$. In particular, it is a non-redundant constraint.
    Therefore, $e_{j}(P)$ will never be degenerate again once $p_i(P)$ is increased further.

    We now argue that \ref{lem:extensionCASE3} and \ref{lem:extensionCASE4} hold by proving the two following claims.

    \begin{claim}
        After the extension process, $\OPT$ lies on $\calg_{2d}$.
    \end{claim}

    \begin{claimproof}
        We proceed by induction in the number of extensions, i.e., show that $\OPT$ lies on $\calg_{k}$ once it has been extended in $k$ directions.

        Assume now that $\OPT$ has been extended in $i-1$ directions, where $i \in \{1, \dots, 2d\}$.
        We focus on the extension in the $i$-th direction of a single $P\in \OPT$.
        Assume first that no polygon has been extended in the $i$-th direction yet, meaning that $P' \in \calg_{i-1}$ for every $P' \in \OPT$ by assumption.
        We show that at any event \ref{case:touchPolygon}, \ref{case:reachPolygon} or \ref{case:degenerate}, $P \in \calg_i$ is maintained.
        Assume $i \notin D(P)$ and let $i_1$ and $i_2$ be the indices of the two non-degenerate edges incident to $e_i(P)$ (w.l.o.g.\ $h(e_{i_1}(P)) = t(e_i(P))$, $t(e_{i_2}(P)) = h(e_i(P))$ and $i_1 < i < i_2$). Notice during the continuous extension of $P$ in the $i$-th direction, $e_{i_1}$ and $e_{i_2}$ get prolonged and all the edges $e_{i'}(P)$, for $i' \notin \{i, i_1, i_2\} \cup D(P)$, are not altered.
        
        At \ref{case:touchPolygon}, assume that $e_{j+d}(\tilde P)$ touches $t(e_i(P))$ (the argument is identical if it touches $h(e_i(P))$) and therefore $e_j(P)$ as well. 
        This means that $i_1 < j < i$.
        In particular, once we have set $p_{j}(P) = -p_{j+d}(\tilde P)$ and continue extending $P$, $e_{j}(P)$ is now a non-degenerate edge incident to $e_i(P)$ while $e_{i_1}$ does not change any longer.
        We have $e_{i_1} \in \call_{i-1}$ since it has only be prolonged since the beginning of the extension, and $e_{j}(P) \in \call_{i-1}$ since it is collinear with $e_{j+d}(\tilde P)$.
        Notice that $t(e_{j}(P)) = h(e_{i_1}(P)) \in \calv_{i-1}$.
        Now substitute $i_1=j$ and repeat the argument for the upcoming events.

        When we arrive at \ref{case:reachPolygon}, all edges of $P$ lie on $\call_{i-1}$, since $e_i(P)$ is collinear with $e_{i+d}(\hat P) \in \call_{i-1}$. Every vertex of $P$ lies on $\calv_{i-1}$.

        At \ref{case:degenerate}, we have $e_{i'}(P) \in \call_{i-1}$ for every $i' \neq i$, and every vertex of $P$ lies on $\calv_{i-1}$.
        This implies $e_i(P) \in \call_i$.

        Assume now $P$ that is not the first polygon to get extended in the $i$-th direction.
        If \ref{case:touchPolygon} occurs and $\tilde P$ has already been extended in the $i$-th direction, since $j \notin \{i, i+d\}$, we have $e_{j+d}(\tilde P) \in \call_{i-1}$ and thus $e_{j}(P) \in \call_{i-1}$ as well.
        Similarly, at \ref{case:reachPolygon}, we have $e_{i+d}(\hat P) \in \call_{i-1}$ (for any $P' \in \OPT$, $e_{i+d}(P')$ stays intact since it is parallel to $e_i(P')$, i.e., $i_1 \neq i+d$, $i_2 \neq i+d$), so $e_{i}(P) \in \call_{i-1}$ and every vertex of $P$ lies on $\calv_i$.
        For \ref{case:degenerate}, there is nothing more to show since is does not establish contact of $P$ with a new polygon.
    \end{claimproof}

    \begin{claim}
        After the extension process, $\OPT$ satisfies \ref{lem:extensionCASE4}.
    \end{claim}

    \begin{claimproof}
        Clearly, for any $P \in \OPT$, $i \in [2d]$, immediately after the extension of $P$ in the $i$-th direction (assuming $i \notin D(P)$ before the extension), $e_i(P)$ satisfies \ref{lem:extensionCASE4}.
        If the extension of $P$ in the $i$-th direction has ended with \ref{case:reachPolygon}, then increasing $p_i(P)$ makes $P$ and $\hat P$ intersect. This means that $e_i(P)$ is not degenerate, and $\interior(e_i(P))$ is non-empty.
        In particular, $\interior(e_i(P))$ touches $\hat P$ and satisfies \ref{lem:extensionCASE4} at the end of the process.
        If the extension of $P$ in the $i$-th direction has ended with \ref{case:degenerate}, then either $e_i(P)$ stays degenerate until the end (so \ref{lem:extensionCASE4} is trivially satisfied) or it becomes non-degenerate during the extension of $P$ in another direction than $i$.
        The latter case occurs at the event \ref{case:touchPolygon} during the extension of $P$ in some direction $i' \in [2d]\setminus \{i, i+d\}$, when $e_{i'}(P)$ touches some $\tilde P \in \OPT$ on the edge $e_{i+d}(\tilde P)$. Then, the algorithm sets $p_i(P) = -p_{i+d}(\tilde P)$ and continues the extension of $P$ in the $i'$-th direction, which prolongates $e_i(P)$. As argued above, at the end of the entire extension process, $e_i(P)$ is non-degenerate, and $\interior(e_i(P))$ touches $\tilde P$, thus satisfying \ref{lem:extensionCASE4}.
    \end{claimproof}
\end{proof}

\section{Proof of the partitioning lemma}
\label{sec:paritioning2}

\partitioningTrue*

\begin{proof}
Let $C$ be a structured container with boundary $\partial C=s_1 f_1 s_2 f_2 \cdots s_\kappa f_\kappa$, where $s_1,\dots, s_{\kappa}$ are cutting lines and $f_1, \dots, f_{\kappa}$ are  fences (to facilitate the proof, we write  $s_{\kappa+1} = s_1$ and $f_{\kappa+1} = f_1$).
Moreover, for some index $\kappa'$ with $\kappa' < \kappa$, the cutting lines $s_1,\dots,s_{\kappa'}$ are left cutting lines and $s_{\kappa'+1},\dots,s_5$ are right cutting lines.
We also assume that $\kappa$ is minimal\mnote{is this condition necessary for the proof? Because I think it is not the case in the figures, so I would probably remove this additionnal assumption \ant{it is used multiple times actually}}, i.e., $f_{j} \cup s_{j+1} \cup f_{j+1}$ never forms a fence (otherwise we merge these curves onto one fence $\tilde f$ and $s_1 f_1 \cdots, s_j \tilde f s_{j+2} \cdots f_\kappa$ is still the boundary of $C$ but the index $\kappa$ decreases).
Without loss of generality, we assume that $\kappa'\ge \lceil \frac{\kappa}{2} \rceil$, that is, we assume that there are not less left cutting lines than right cutting lines, otherwise we mirror the instance vertically.
We first treat the case when $\kappa = 5$. In particular, we have $\kappa'\ge 3$. 
We will address the proof of cases with $\kappa \leq 4$ at the end of the proof since it suffices to adapt the arguments made for $\kappa = 5$.

Without loss of generality, we assume that $C$ contains no spurs. A \emph{spur} is a vertex whose two incident edges overlap, this might happen if some line segment $s_j$ is degenerate and the two edges in $f_j$ and $f_{j+1}$ that are incident to $s_j$ overlap. If the edges $\overline{w'w}$ and $\overline{ww''}$ form a spur around $w$, then replace those two edges with the new edge $\overline{w'w''}$ and reiterate this procedure as long as there is a spur. Removing spurs this way does not decrease the interior of a container nor create any new crossing.
In particular, any polygon (with non-crossing, positively oriented boundary) has non-empty interior if and only if its boundary is non-degenerate once all spurs are removed.

Suppose there exists a polygon $P_0$ in $C$ protected from the left via $s_2$; 
take such $P_0$ with the right-most right edge (breaking ties arbitrarily). 
Since $P_0$ is protected (see Definition~\ref{def:protected_polygon}) from the left via $s_2$, there exist a top-fence $\gamma_h$ emerging from $s_2$ with endpoint $h = h(e_1(P_0))$ and a bottom-fence $\gamma_t$ emerging from $s_2$ with endpoint $t = t(e_1(P_0))$. 

Now, let $p_h$ and $p_t$ be the two points that satisfy the following conditions: 
\begin{enumerate}[label = (\roman*)]
    \item\label{first} $p_h$ and $p_t$ are in the same vertical line as $h$ and $t$ such that $p_h$ is above or coincides with $h$ and $p_t$ is below or coincides with $t$.
    \item\label{second} $p_h$ (resp. $p_t$) lies either on $\partial C$ or on $\partial P_h \setminus (e_1(P_h) \cup e_{d+1}(P_h))$ (resp.\ $\partial P_t \setminus (e_1(P_t) \cup e_{d+1}(P_t))$), where $P_h\in\OPT(C)$ (resp.\ $P_t\in\OPT(C)$) is protected in $C$.
    \item\label{third} The semi-open vertical line segments $\overline{hp_h} \setminus \{h\}$ and $\overline{tp_t} \setminus \{t\}$ do not contain any point that satisfies condition~\ref{second}. 
\end{enumerate}
It is clear that these two points exist and are unique. 
We will construct $\Gamma$ by case analysis, depending on the location of $p_h$ and $p_t$.


\begin{figure}%
\centering
\includegraphics[width=0.7\textwidth]{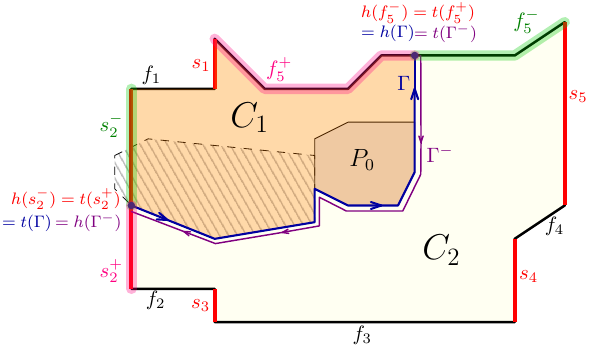}%
\caption{Illustration of the notation used to construct $C_1$ and $C_2$ according to the location of $p_h$ and $p_t$. In this example, $C_1$ has boundary $\partial C_1 = \Gamma f_5^+ s_1 f_1 s_2^-$ and $C_2$ has boundary $\partial C_2 = s_2^+ f_2 s_3 f_3 s_4 f_4 s_5 f_5^- \Gamma^-$.}%
\label{fig:partitioningStructure}%
\end{figure}

Before we proceed,
let us first introduce some notation that allows us to define the boundaries of $C_1$ and $C_2$ in a compact form, see Figure~\ref{fig:partitioningStructure}.
For $i \in [\kappa]$, $r_i \in \{s_i, f_i\}$ and $c \in \{h(\Gamma), t(\Gamma)\}$ (as it will be defined below) such that $c \in r_i$, we define $r_i^+$ to be the curve on $r_i$ from $c$ to $h(r_i)$ and $r_i^-$ to be the curve on $r_i$ from $t(r_i)$ to $c$ ($r_i^+$ or $r_i^-$ might only consist of a singleton point). In particular, $r_i = r_i^-r_i^+$ and $t(r_i^-) = t(r_i)$, $h(r_i^-) = c = t(r_i^+)$ and $h(r_i^+) = h(r_i)$.\mnote{this sentence is a bit difficult to sparse. I guess it would be better to define this later, the first time that we use it. Maybe also refer to a figure when we use these notaions + and - }
For the oriented curve $\Gamma$ that we will define below, let $\Gamma^-$ be $\Gamma$ but with inverse orientation (i.e., $h(\Gamma^-) = t(\Gamma)$ and $t(\Gamma^-) = h(\Gamma)$).

We first cover the simplest cases, namely when $P_0$ exists and $\interior(e_1(P_0)) \subseteq \interior(C)$, which are illustrated in Figure~\ref{fig:partitioning_others}.

\begin{enumerate}[label=\textbf{(A\arabic*)}, wide, labelwidth=0pt, labelindent=0pt]
\setlength{\itemsep}{5pt}
    \item \label{case:f1f2}
    \textbf{Both $p_h$ and $p_t$ lie on $f_1 \cup f_2$.}
    This implies $p_h \in f_1$ and $p_t \in f_2$ since $f_1$ lies above $f_2$.%
    \footnote{For two objects $A,B \subseteq \R^2$, we say that \emph{$A$ lies above (below) $B$} if for every $(a,b) \in A \times B$ such that $a$ and $b$ have the same $x$-coordinate, $a$ has larger (smaller) $y$-value than $b$. We define \emph{lying on the left} or \emph{right} analogously.}
    Choose $\Gamma = \overline{p_h p_t}$, thus $\partial C_1 = \Gamma f_2^+ s_3 \cdots s_1 f_1^-$ and $\partial C_2 = f_1^+ s_2 f_2^- \Gamma^-$.
    %
    \item \label{case:f3f4f5}
    \textbf{$p_h$ or $p_t$ (or both) lies on $\bigcup_{i \geq 3} f_i$.}
    Assume that $p_t \in f_j$, $j \geq 3$ (if $p_t$ lies on multiple fences, choose $j$ to be the smallest index).
    Choose $\Gamma = \gamma_h \overline{h p_t}$, thus $\partial C_1 = \Gamma f_j^+ s_{j+1} \cdots, f_1 s_2^-$ and $\partial C_2 = s_2^+ f_2 \cdots f_j^- \Gamma^-$.
    \item \label{case:protectedRight}
    \textbf{$P_h$ or $P_t$ (or both) exists and is protected from the right in $C$.}
    Assume that $P_h$ is protected from the right.
    Let $\delta_h$ be a bottom-fence starting at $p_h$ and ending on some right cutting line $s_j$ that protects $P_h$.
    Choose $\Gamma = \gamma_h \overline{h p_h} \delta_h$, thus $\partial C_1 = \Gamma s_j^+ f_j \cdots f_1 s_2^-$ and $\partial C_2 = s_2^+ f_2 \cdots, f_{j-1} s_j^- \Gamma^-$.
    \item \label{case:protectedLeftLeft}
    \textbf{Both $P_h$ and $P_t$ exist and are protected from the left in $C$.}
    From condition~\ref{second}, there exists a bottom-fence $\delta_h$ from $s_1$ to $p_h$ which protects $P_h$ and terminating in $p_h$
    and a top-fence $\delta_t$ from $s_j$ ($3 \leq j \leq \kappa'$) to $p_t$ which protects $P_t$.
    Choose $\Gamma = \delta_h  \overline{p_h p_t} \delta_t$, thus $\partial C_1 = \Gamma s_j^+ f_j \cdots f_5 s_1^-$ and $\partial C_2 = s_1^+ f_1 s_2 \cdots s_j^- \Gamma^-$.
    \item \label{case:f1protectedLeft}
    \textbf{$P_t$ exists and is protected from the left in $C$ and $p_h \in f_1$; or $P_h$ exists and is protected from the left in $C$ and $p_t \in f_2$.}
    Without loss of generality, assume the first situation.
    Let $\delta_t$ be the top-fence from $s_j$ ($3 \leq j \leq \kappa'$) to $p_t$ which protects $P_t$.
    Choose $\Gamma = \overline{p_h p_t} \delta_t$, thus $\partial C_1 = \Gamma s_j^+ f_j \cdots s_1 f_1^-$ and $\partial C_2 = f_1^+ s_2 \cdots f_{j-1} s_j^- \Gamma^-$.
\end{enumerate}

\begin{figure}%
    \centering
    \subfloat[\centering Case \ref{case:f1f2}.]
    {\includegraphics[width=0.46\textwidth, page=6]{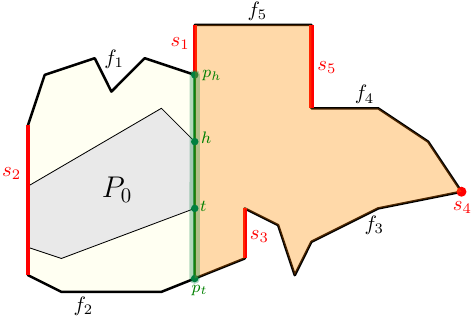} }%
    \qquad
    \subfloat[\centering Case \ref{case:f3f4f5}.]{{\includegraphics[width=0.46\textwidth, page=7]{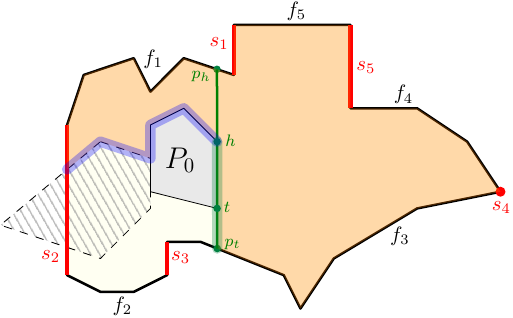} }}%
    \\
    \vspace{30pt}
    \subfloat[\centering Case \ref{case:protectedRight}.]
    {\includegraphics[width=0.46\textwidth, page=8]{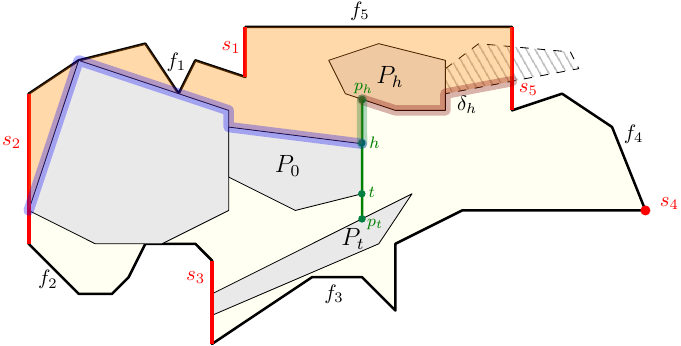} }%
    \qquad
    \subfloat[\centering Case \ref{case:protectedLeftLeft}.]{{\includegraphics[width=0.46\textwidth, page=9]{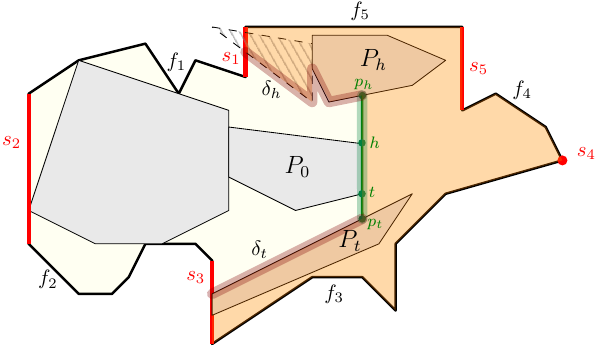} }}%
    \\
    \vspace{30pt}
    \subfloat[\centering Case \ref{case:f1protectedLeft}.]{{\includegraphics[width=0.46\textwidth, page=10]{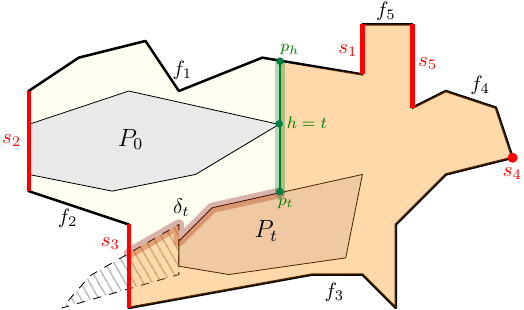} }}%
    \caption{
    Bipartitions of $C$ for $\interior(e_1(P_0)) \subseteq \interior(C)$.
    The areas $\interior(C_1)$ (orange) and $\interior(C_2)$ (yellow), are separated by the highlighted curve $\Gamma$
    ($\gamma_h, \gamma_t$ in blue, $\ell$ in green, $\delta_h, \delta_t$ in orange).
    }
    \label{fig:partitioning_others}
\end{figure}

Now, we discuss the instances where $P_0$ exists but $\interior(e_1(P_0)) \not\subseteq \interior(C)$, or equivalently, $\interior(e_1(P_0)) \cap \partial C \neq \emptyset$, see Figure~\ref{fig:partitioning_touchBoundary}. For \ref{case:sj}, $\Gamma$ can be constructed straight away, while \ref{case:intersectFence} reduces to one of \ref{case:f1f2} to \ref{case:f1protectedLeft} by redefining $p_h$ and $p_t$ as an intermediate step.

\begin{enumerate}[label=\textbf{(B\arabic*)}, wide, labelwidth=0pt, labelindent=0pt]
\setlength{\itemsep}{5pt}
    \item \label{case:sj}
    \textbf{$\overline{p_hp_t}$ overlaps a right cutting line $s_j$.}
    Assume $p_h \in s_j$.
    Choose $\Gamma = \gamma_h$, thus $\partial C_1 = \Gamma s_j^+ f_j \cdots s_1 f_1 s_2^-$ and $\partial C_2 = s_2^+ f_2 \cdots s_j^- \Gamma^-$.
    \item \label{case:intersectFence}
    \textbf{$\interior(e_1(P_0))$ intersects a fence $f_j$.} \mnote{maybe this could be the first case, even before defining $p_h$ and so ? instead of re-defining them temporarily}
    \anote{I don't know... the point is precisely that it reduces to the previous cases by making temporary re-definitions and is not really a case of its own, just an intermediate step}
    Assume that $f_j$ is oriented from left to right.
    Therefore $t \in f_j$, and $e_1(P_0)$ intersects the interior of the vertical line $\tilde e$ of a top-fence on $\partial C$, thus $\tilde e = e_{d+1}(\tilde P)$ for some $\tilde P \in \OPT \setminus \OPT(C)$; furthermore $j \notin \{1,5\}$.
    We redefine $p_h$ and $p_h$:
    if $j \neq 2$, then set $p_t$ and $p_h$ as $t(\tilde e)$ or $h$, whichever lies lower. and then construct $\Gamma$ as in \ref{case:f3f4f5}, i.e., $\Gamma = \gamma_h \overline{hp_t}$ and $\partial C_1$, $\partial C_2$ accordingly.
    For $j=2$,
    set $p_t = t(\tilde e)$, and then
    recompute $p_h$ according to \ref{first}, \ref{second} and \ref{third} and additionally such that any point in the semi-open vertical line segment $\overline{t(\tilde e)p_h} \setminus \{t(\tilde e)\}$ satisfies \ref{second} (so $p_h$ lies above $f_2$).
    Then, we apply the case analysis and construction of the cases \ref{case:f1f2} to \ref{case:f1protectedLeft}.
    Notice that in any case, we have $\interior(\overline{p_hp_t}) \subseteq \interior(C)$.
\end{enumerate}

\begin{figure}%
    \centering
    \subfloat[\centering Case \ref{case:sj}.]
    {\includegraphics[width=0.52\textwidth, page=5]{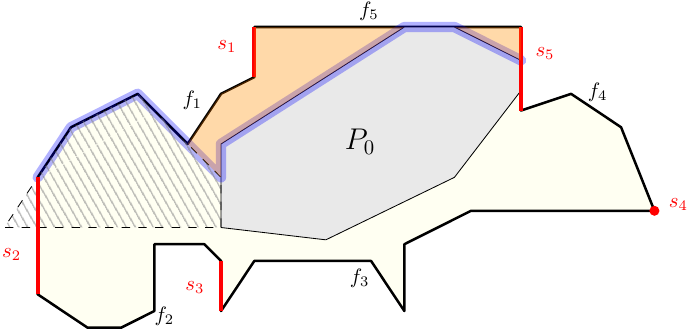} }%
    \qquad
    \subfloat[\centering Case \ref{case:intersectFence}.]{{\includegraphics[width=0.35\textwidth, page=4]{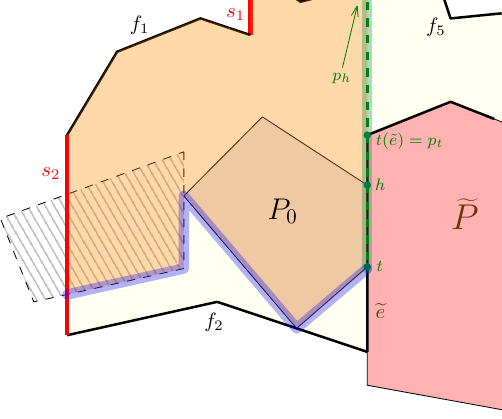} }}%
    \caption{
    Bipartitions of $C$ if the right side of $P_0$ touches $\partial C$. $\interior(C_1)$ is orange, $\interior(C_2)$ is yellow.
    }
    \label{fig:partitioning_touchBoundary}
\end{figure}

We now consider instances where $P_0$ does not exist, see Figure~\ref{fig:partitioning_noP0}.
If $s_1$ and $s_3$ both lie on the right of $s_2$, then it suffices to ``shift'' $s_2$ a little to the right, as defined in \ref{case:noP0rr}.
Otherwise, i.e., in \ref{case:noP0l}, we make $\ell$ going down (or up) starting from the bottom (resp.\ top) of $s_2$ and reduce it to simpler cases, depending on the location of $p_h$ and $p_t$.
There might also occur another very particular case, namely if there is one or two polygons that touches $s_2$ on its top and/or bottom (so $\ell$ is degenerate) and is protected from the left by $f_2$ resp.\ $f_1$, so there is no ``gap'' above resp.\ below the polygon(s).
We have to handle this case individually.

Let us describe the latter case more formally.
Assume that $s_3$ lies on the left of $s_2$, and let $P_3$ be the polygon in $\OPT(C)$ (if it exists) with
\anote{inline or regular enumerate?}
\begin{enumerate*}
    \item $h(s_2)$ lies on the bottom (but not on the side) of $P_3$
    \item $P_3$ is protected from the left via $s_3$
    \item for the top-fence $\mu_3$ protecting $P_3$, we have $f_2 \subseteq \mu_3$.
\end{enumerate*}
If $s_1$ lies on the left of $s_2$, the polygon $P_1 \in \OPT(C)$, if it exists, touches $t(s_2)$ on its bottom (and the bottom-fence $\mu_1 \supseteq f_1$ protecting $P_1$), is defined symmetrically.

\begin{enumerate}[label=\textbf{(C\arabic*)}, wide, labelwidth=0pt, labelindent=0pt]
\setlength{\itemsep}{5pt}
    \item \label{case:noP0rr}
    \textbf{$P_0$ does not exist and $s_1$ and $s_2$ both lie on the right of $s_2$.}
    Choose $\Gamma$ to be any vertical line segment $\ell$ (on $\calg_{2d}$) with $t(\ell) \in f_1$ and $h(\ell) \in f_2$ which lies strictly on the right of $s_2$ and that does not intersect any polygon in $\OPT(C)$. If $\ell$ does not exist, then the leftmost $P \in \OPT$ between $f_1$ and $f_2$ must touch $s_2$ and is therefore protected, which contradicts the nonexistence of $P_0$. Then $C_1$ and $C_2$ are constructed as in \ref{case:f1f2}, i.e., $\partial C_1 = \Gamma f_2^+ s_3 \cdots s_1 f_1^-$ and $\partial C_2 = f_1^+ s_2 f_2^- \Gamma^-$.
    \item \label{case:noP0lP1P3}
    \textbf{$P_0$ does not exist and either: $s_1$ lies on the right of $s_2$ and $P_3$ exists, or $s_3$ lies on the right of $s_2$ and $P_1$ exists, or both $P_1$ and $P_3$ exist.}
    Let us first consider the former occurrence, i.e., when $s_1$ lies on the right of $s_2$ and $P_3$ exists (the construction for the second one is analogous).
    We assume that $s_2$ is non-degenerate or that the segment of $\mu_3$ on the right of $s_2$ does not form a spur with $f_1$. Indeed, otherwise $s_2$ can be dragged at the right extreme point of the spur.
    Let $\ell$ be any vertical line segment (on $\calg_{2d}$) strictly on the right of $s_2$ with $t(\ell) \in f_1$ and $h(\ell) \in \mu_3$ that does not intersect any polygon in $\OPT(C)$. $\ell$ exists by the same argument as in \ref{case:noP0rr}.
    Let $\mu_3'$ be the segment of $\mu_3$ from $h(\ell)$ to $h(s_2)$ and choose $\Gamma = \ell\mu_3'$, thus $\partial C_1 = \Gamma f_2 s_3 \cdots s_1 f_1^-$ and $\partial C_2 = f_1^+ s_2 \Gamma^-$.
    
    We construct $\Gamma$ similarly when both $P_1$ and $P_3$ exist.
    Again, we assume that $s_2$ is non-degenerate or that the segments of $\mu_1$ and $\mu_3$ on the right of $s_2$ do not form a spur, since otherwise $s_2$ can be dragged at the right extreme point of that spur.
    Let $\ell$ be any vertical line segment (on $\calg_{2d}$) strictly on the right of $s_2$ with $t(\ell) \in \mu_1$ and $h(\ell) \in \mu_3$ that does not intersect any polygon in $\OPT(C)$.
    Let $\mu_3'$ defined as above and $\mu_1'$ be the segment of $\mu_1$ from $h(t(s_2))$ to $t(\ell)$ and choose $\Gamma = \mu_1' \ell \mu_3'$, thus $\partial C_1 = \Gamma f_2 s_3 \cdots s_1 f_1$ and $\partial C_2 = s_2 \Gamma^-$.
    %
    \item \label{case:noP0l}
    \textbf{$P_0$ does not exist and $s_1$ or $s_3$ (or both) lie on the left of $s_2$, and without \ref{case:noP0lP1P3}}
    Assume $s_3$ lies on the left of $s_2$, i.e., $f_2$ is oriented from right to left. Define $h = t = p_h = h(s_2)$, $\gamma_h = \gamma_t = \{h(s_2)\}$ and $p_t$ according to \ref{first}, \ref{second}, and 
    such that any point on the semi-open vertical line segment $\overline{h(s_2)p_t} \setminus \{h(s_2)\}$ satisfies \ref{second} (so $p_t$ lies below $h(s_2)$).
    If now $p_t \in \bigcup_{i\geq3}f_i$ or $P_t$ is protected from the right, proceed as in \ref{case:f3f4f5} or \ref{case:protectedRight}, respectively.
    In the case that $P_t$ is protected from the left, let $\delta_t$ be the top-fence which protects $P_t$ (from the left cutting line $s_j$, $3 \leq j \leq \kappa'$ to $p_t$).
    Then choose $\Gamma = \overline{p_hp_t}\delta_t$ and therefore $\partial C_1 = \Gamma s_j^+ f_j \cdots f_1 s_1$ and $\partial C_2 = f_2 s_2 \dots s_j^- \Gamma^-$ (this case is similar to \ref{case:f1protectedLeft}).
\end{enumerate}

\begin{figure}%
    \centering
    \subfloat[\centering Case \ref{case:noP0rr}.]{{\includegraphics[width=0.40\textwidth, page=2]{Figures_PDF/fig_partitioning.pdf} }}%
    \qquad
    \subfloat[\centering Case \ref{case:noP0l}.]
    {\includegraphics[width=0.40\textwidth, page=1]{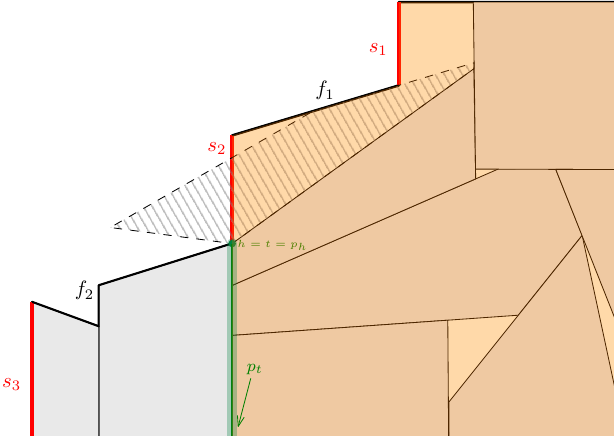} }%
    \\
    \vspace{30pt}
    \subfloat[\centering Case \ref{case:noP0lP1P3}, when $P_3$ exists and $s_1$ lies on the right of $s_2$.]{{\includegraphics[width=0.40\textwidth, page=1]{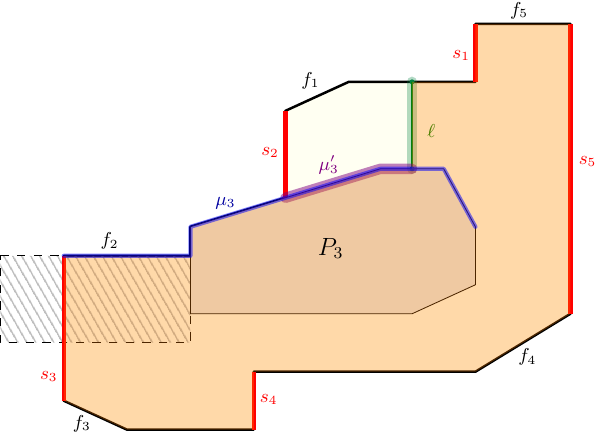} }}%
    \qquad
    \subfloat[\centering Case \ref{case:noP0lP1P3}, when both $P_1$ and $P_3$ exist.]
    {\includegraphics[width=0.40\textwidth, page=2]{Figures_PDF/fig_partitioningSpecial.pdf}}%
    \caption{
    Bipartitions of $C$ if $P_0$ does not exist. The cutting line $\ell$ (green) lies on the right of  $\interior(C_1)$ and on the left of $\interior(C_2)$.
    }
    \label{fig:partitioning_noP0}
\end{figure}

For every case, it is clear that $\ell = \Gamma \cap \overline{p_hp_t}$ does not intersect any protected polygon by the criteria \ref{first}, \ref{second} and \ref{third}.
For every case, it is also easy to verify that $\Gamma \setminus \ell$ is a union of fences. Since fences do not intersect any polygon in $\OPT$ by definition,
\ref{strongC} and \ref{strongD} follow from these observations.
Notice also that \ref{stronE} follows from \ref{strongB} and \ref{strongD}: For any $P \in \Pr(C)$ that is protected in $C$ by fences $\sigma_h, \sigma_t$ via $s$, we have $P \in \OPT(C_1)$ or $P \in \OPT(C_2)$ by \ref{strongB} and \ref{strongD}. Assuming the former, then either $\ell$ crosses the fences protecting $P$, so then the segments of $\sigma_h, \sigma_t$ starting from $\ell$ form a new pair of fences in $C_1$ that protect $P$ in $C_1$ via $\ell$, or otherwise, $\sigma_h, \sigma_t$ lie on $C_1$ and protect $P$ in $C_1$ via $s$ as before.
We have $P \in \Pr(C_1)$ in both cases.
Therefore, to finish the proof (for $\kappa = 5$), we need to show \ref{strongB}.

It is simple to verify that the polygons $C_1, C_2$ constructed above are structured, i.e., that their boundaries consist each of at most five disjoint vertical cutting lines and five fences that alternate, for every case treated above.
It remains to show that $C_1$ and $C_2$ are indeed containers, i.e., that they have non-empty interior and no crossings, and that they partition the interior of $C$.

\begin{claim}\label{claim:noSpurs}
    Both $C_1$ and $C_2$ as constructed by Lemma~\ref{lem:partitioningStronger} have non-empty interior.
\end{claim}

\begin{claimproof}
    We can argue either by proving $\interior(C_i) \neq \emptyset$ directly or by proving that $\partial C_1$ and $\partial C_2$ are non-degenerate curves once all spurs are deleted.
    Since $\partial C$ and $\Gamma$ by themselves contain no spurs (by assumption and construction, respectively), for every construction of $\Gamma$, it suffices to show either $\Gamma \setminus \partial C \neq \emptyset$, or $\partial C_1 \setminus \Gamma \neq \emptyset$ and $\partial C_2 \setminus \Gamma \neq \emptyset$.
    It is simple to verify that $\overline{p_h p_t}$ does not form any spur with $\partial C$, so if $\ell$ is non-degenerate, the claim holds as well.
    We argue case by case.
    \begin{description}
        \item[\ref{case:f1f2}.] \enote{Why the dot? \ant{what do you prefer?}}We have $\interior(\Gamma) = \interior(\ell) \subseteq \interior(C)$, which excludes any spur.
        \item[\ref{case:f3f4f5}.] Assume $p_t \in f_j$, $j\geq 3$. If $\partial C_1 \setminus \Gamma = \emptyset$, we have $\Gamma^- = f_j^+ s_{j+1} \cdots, f_1 s_2^-$, and thus the contradiction that $f_5 s_1 f_1$ forms a fence.
        We have $\interior(C_i) \neq \emptyset$ since $P_0 \in \interior(C_2)$.
        \item[\ref{case:protectedRight}.] Assume $P_h$ is protected from the right. We have $\partial C_2 \setminus \Gamma \neq \emptyset$ because $\gamma_h$ lies (strictly) above $\gamma_t$, which itself lies above $f_2$, so $\gamma_h$ cannot form any spur with $f_2$.
        Similarly, we have $\partial C_1 \setminus \Gamma \neq \emptyset$ because $\delta_h$ cannot form a spur with $f_j$, since $\delta_h$ (which is a bottom-fence) lies (strictly) below $f_j$.
        \item[\ref{case:protectedLeftLeft}.] We have $P_t \in \interior(C_1)$ and $P_0 \in \interior(P_0)$.
        \item[\ref{case:f1protectedLeft}.] Same argument as for \ref{case:protectedLeftLeft}.
        \item[\ref{case:sj}.] We have $\Gamma = \gamma_h \subseteq \partial C$ if and only if $s_j^+$ and $s_2^-$ are singletons and $\Gamma^- = f_j \cdots s_1 f_1$, which contradicts the assumption that $f_5s_1f_1$ cannot form a fence, thus $\Gamma \setminus \partial C \neq \emptyset$.
        \item[\ref{case:intersectFence}.] The argument reduces to the cases \ref{case:f1f2} to \ref{case:f1protectedLeft}.
        \item[\ref{case:noP0rr}.] Same argument as for \ref{case:f1f2}.
        \item[\ref{case:noP0lP1P3}.] The claim follows from the assumption that $f_1$ and $\mu_3'$ or $\mu_1'$ and $\mu_3'$ do not form any spur.
        \item[\ref{case:noP0l}.] We can assume $p_t = p_h$, otherwise $\ell$ is non-degenerate.
        If $p_t \in f_j$, $j \geq 3$, then $\Gamma = \{h(s_2)\}$, so $\partial C_1 \setminus \Gamma \neq \emptyset$ and $\partial C_2 \setminus \Gamma \neq \emptyset$ trivially hold.
        Else, if $P_t$ exists (i.e., $\Gamma = \delta_t$), then $\partial C_1\setminus \Gamma \neq \emptyset$ because $s_2 \in \partial C_1$, and $\partial C_2 \setminus \Gamma \neq \emptyset$ because $f_2 \setminus \Gamma \neq \emptyset$ (either because $P_t$ is protected from the right or because otherwise we are in \ref{case:noP0lP1P3}).
    \end{description}
\end{claimproof}

\begin{claim}\label{claim:weaklySimple}
    The respective boundaries of $C_1$ and $C_2$ contain no crossings. 
\end{claim}

\begin{claimproof}
    For every case, the claim follows from one or many of those statements:
    \begin{enumerate}[label={(\arabic*)}, leftmargin=*]
        \item \label{crossproof1} There is no crossing on $\Gamma$.
        \item \label{crossproof2} There are also no crossings on $\partial C_1 \setminus \Gamma$ and there are also no crossings on $\partial C_2 \setminus \Gamma$.
        \item \label{crossproof3} $\partial C$ and $\ell$ do not cross.
        \item \label{crossproof4} Any fence does in $C$ does not cross $\partial C$.
    \end{enumerate}
    We now prove those statements.
    
    For all cases but \ref{case:protectedLeftLeft} and \ref{case:noP0lP1P3} (when both $P_1$ and $P_3$ exist), $\Gamma$ is horizontally monotonic, so \ref{crossproof1} is trivially true.
    For the two latter cases, the statement is true because $\delta_h$ lies above $\delta_t$ in \ref{case:protectedLeftLeft} and $P_1$ lies above $P_3$ in \ref{case:noP0lP1P3}.

    \ref{crossproof2} holds because if there are two segments of $\partial C_1 \setminus \Gamma$ (or $\partial C_2 \setminus \Gamma$) that form a crossing, then those two segments also lie and form a crossing on $\partial C$, which contradicts the assumption that $C$ is a structured container.

    \ref{crossproof3} holds because $\interior(\ell) \subseteq \interior(C)$.

    To show \ref{crossproof4}, let $\gamma$ be a fence in $C$.
    If $\gamma$ lies only on the top or bottom of one polygon $P \in \OPT(C)$ then a crossing between $\gamma$ and $\partial C$ implies that $\partial C$ to intersect $P$. However, $P$ is protected in $C$ by $\gamma$, which is a contradiction.
    If $\gamma$ lies on top or bottom of two polygons $P, P'$ such that $P$ sees $P'$ and one of them is protected in $C$, then $\gamma$ and $\partial C$ must overlap on $e_1(P) \cap e_{d+1}(P')$. However, this implies that $P$ lies inside $C$ and $P'$ lies fully outside $C$ or vice-versa, which is not possible: either both $P$ and $P'$ lie inside $C$ or one of them is intersected by a cutting line of $C$.

\end{claimproof}

\begin{claim}\label{claim:partition}
    The polygons $C_1$ and $C_2$ partition $C$.
\end{claim}

\begin{claimproof}
    For any $C' \in \{C, C_1, C_2\}$ and any point $w \notin \partial C$, let $g_w$ be the vertical line with $w \in g_w$, let $w_{C'}^h$ ($w_{C'}^t$) be the lowest (highest) point on $\partial C' \cap g_w$ strictly above (below) $w$.
    Let also $f_{C'}^h$ ($f_{C'}^t$) be a fence on $\partial C'$ with $w_{C'}^h \in f_{C'}^h$ ($w_{C'}^t \in f_{C'}^t$).
    Notice that for every $w \in \interior(C)$, $w_{C}^h$ and $w_{C}^t$ exist and $f_{C}^h$ ($f_{C}^t$) is oriented from right to left (left to right).
    We also have $\interior(\overline{w_{C'}^h w_{C'}^t}) \subseteq \interior(C')$.
    We now argue that for every point $w\in \interior(C)\setminus \Gamma$, we have that
    \begin{enumerate}[label=(\arabic*), leftmargin=*]
        \item\label{part:whwtexist} $w_{C'}^h$ and $w_{C'}^t$ exist.
        \item\label{part:orientation} $f_{C'}^h$ is oriented from right to left and $f_{C'}^t$ is oriented from left to right.
    \end{enumerate}
    for either $C' = C_1$ or $C'=C_2$ (but not both).
    Notice that since $\partial C_1 \cup \partial C_2 = \partial C \cup \Gamma$ by construction, for any point $w\in \interior(C) \setminus \Gamma$, $w_{C'}^h$ and $w_{C'}^t$ must exist for either $C'=C_1$ or $C'= C_2$.
    We argue case by case.
    \begin{description}
        \item[\ref{case:f1f2}.] For any $w \in \interior(C) \setminus \Gamma$, we have $g_w \cap \Gamma = \emptyset$.
        Notice also that $f_1$ is oriented from right to left and $f_2$ is oriented from left to right.
        In particular, for any $w \in \interior(C) \setminus \Gamma$, we have \ref{part:whwtexist} and \ref{part:orientation} for $C' = C_2$ if and only if $w_C^h \in f_1$, $w_C^t \in f_2$ and $w$ lies on the left of $\Gamma$, and \ref{part:whwtexist} and \ref{part:orientation} for $C' = C_1$ otherwise.
        \item[\ref{case:f3f4f5}.] Assume $p_t \in f_j$, $j\geq 3$, so $f_j$ is oriented from left to right.
        Notice that $t(f_j)$ as well as $h(f_i), t(f_i)$ for $2 \leq i \leq j-1$, lie on the left of $\overline{hp_t}$
        and below $\gamma_h$.
        
        First assume that $w_{C_1}^t$ exists, so either $w_{C_1}^t = w_C^t \in f_i$, $i \geq j$ or $i=1$, or  $w_{C_1}^t \in \gamma_h$.
        Either way, $w$ lies above or on the right of $\gamma_h$, and in the latter case, so we have $w_{C_1}^h = w_C^h \in  f_i$ for $i \geq j+1$ or $i=1$. and thus \ref{part:whwtexist}.
        We have that $f_{C_1}^h = f_C^h$ is oriented from right to left (since \ref{part:orientation} must hold for $C'=C$), and $f_{C_1}^t = f_C^t$ is oriented from right to left; for the same reason if $w_{C_1}^t = w_C^t$, and because $f_{C_1}^t = \gamma_h$ is oriented from left to right otherwise.
        This confirms \ref{part:orientation} if  $w_{C_1}^t$ exists.
        
        Assume now that $w_{C_2}^t$ exists, so $w_{C_2}^t = w_C^t \in \bigcup_{2\leq i \leq j} f_i$.
        If $w_C^h \in \partial C_2$, we have $w_{C_2}^h = w_C^h$ and the claim is clear.
        Otherwise, $\interior(\overline{w_{C}^h w_{C}^t})$ crosses $\gamma_h$ so $w_{C_2}^h \in \Gamma^-$ and \ref{part:whwtexist} and \ref{part:orientation} hold for $C' = C_2$, since $\Gamma^-$ is oriented from right to left.
        \item[\ref{case:protectedRight}.] Assume that $P_h$ is protected from the right and without loss of generality $j=5$.
        Notice that $\Gamma$ is horizontally monotone.
        If $w$ lies above $\Gamma$, then $w_{C}^h \in f_1 \cup f_5$, therefore $w_{C_1}^h = w_{C}^h$ and $w_{C_1}^t \in \Gamma$ exist, which confirms \ref{part:whwtexist}.
        Since $\Gamma = \gamma_h$ is oriented from left to right and $f_1$ and $f_5$ are oriented from right to left, \ref{part:orientation} holds as well.
        The same argument can be used if $w$ lies below $\Gamma$.
        \item[\ref{case:protectedLeftLeft}.] Let $j \in [\kappa']$ such that $P_t$ is protected via $s_j$.
        We make a similar observation as in the argument of \ref{case:f3f4f5}:
        $f_j$ is oriented from left to right, $t(f_j)$, $h(f_5)$ as well as $h(f_i), t(f_i)$ for $i \in [j-1]$, lie on the left of $\overline{p_hp_t}$.
        Therefore, we can apply the same idea as for \ref{case:f3f4f5}:
        if $w_{C_1}^t$ exists, we have $w_{C_1}^h = w_C^h \in \bigcup_{i\geq j+1} f_i$ or $w_{C_1}^h \in \delta_t$, which fulfils \ref{part:whwtexist}. \ref{part:orientation} holds because $f_C^h$ and $\delta_t$ are oriented from right to left and $f_C^t$ and $\delta_h$ are oriented from left to right.

        If $w_{C_2}^t$ exists, then either $w_{C}^h$ lies on the left of $s_1$ and so $w_{C_2}^h = w_C^h \in \bigcup_{i \leq j-1}$ or $\interior(\overline{w_{C}^h w_{C}^t})$ crosses $\delta_h$ so $w_{C_2}^h \in \Gamma^-$. Either way, this implies \ref{part:whwtexist}.
        \ref{part:orientation} is then simple to verify.
        \item[\ref{case:f1protectedLeft}.] Same argument as for \ref{case:f3f4f5}.
        \item[\ref{case:sj}.] Same argument as for \ref{case:protectedRight}.
        \item[\ref{case:intersectFence}.] The argument reduces to the cases \ref{case:f1f2} to \ref{case:f1protectedLeft}.
        \item[\ref{case:noP0rr}.] Same argument as for \ref{case:f1f2}.
        \item[\ref{case:noP0lP1P3}.] Assume first that $P_3$ exists and that $s_1$ lies on the right of $s_2$.
        Clearly, for any $w\in \interior(C) \setminus \Gamma$, we have \ref{part:whwtexist} and \ref{part:orientation} for $C'=C_2$ if and only if $w$ lies on the right of $s_2$, on the left of $\ell$ and (strictly) between $f_1$ and $\mu_3'$ and  \ref{part:whwtexist} and \ref{part:orientation} hold for $C'=C_1$ otherwise.
        The argument is similar if both $P_1$ and $P_3$ exist (substitute $f_1$ by $\mu_1'$).
        \item[\ref{case:noP0l}.] The argument reduces to the cases \ref{case:f1f2} to \ref{case:f1protectedLeft}.
    \end{description}
\end{claimproof}

\noindent\textbf{Instances with $\kappa \leq 4$.}
We conclude the proof of the partitioning lemma by discussing instances with $\kappa \leq 4$.
For $\kappa = 4$ and $\kappa = 3$, we adhere to the construction of $\Gamma$ for $\kappa = 5$ above.
The case analysis is identical as for $\kappa = 5$, except that some cases cannot occur (for example, $\kappa' = 2$ rules out \ref{case:protectedLeftLeft}).

For $\kappa = 2$, we define $P_0, h, t, p_h, p_t$ as above, except that $P_0$ is now protected via $s_1$ (since $s_2$ must be a right cutting line).
If we have $\gamma_h = f_1$ and $\gamma_t = f_2$, since $|\OPT(C)| \geq 2$ by assumption, this means that there is a polygon $P_0' \in \OPT(C)$ that touches $s_1$ and sees $P_0$, so we can set $\Gamma = e_{1}(P_0') \cap e_{d+1}(P_0)$ and thus $\partial C_1 = \partial P_0'$ and $\partial C_2 = \partial P_0$.
Otherwise, we construct $\Gamma$ as before (again by substituting $s_2$ with $s_1$ and $s_j$ with $s_2$ for $j\neq 2$).

Lastly, the proofs of Claim~\ref{claim:noSpurs}, Claim~\ref{claim:weaklySimple} and Claim~\ref{claim:partition} are essentially independent of $\kappa$ and can be easily be adapted and verified for instances with $\kappa \leq 4$.
\end{proof}

\end{document}